\documentclass[12pt,draftcls,journal,onecolumn]{IEEEtran}
\usepackage{amsfonts,color,morefloats,pslatex}
\usepackage{amssymb,amsthm, amsmath,latexsym,booktabs}

\newtheorem{theorem}{Theorem}
\newtheorem{lemma}[theorem]{Lemma}
\newtheorem{proposition}[theorem]{Proposition}

\newtheorem{example}[theorem]{Example}

\newtheorem{remark}[theorem]{Remark}

\newcommand{\ord}{{\mathrm{ord}}}

\newcommand{\lcm}{{\mathrm{lcm}}}
\newcommand{\tr}{{\mathrm{Tr}}}

\newcommand{\gf}{{\mathrm{GF}}}

\newcommand{\wt}{{\mathtt{wt}}}

\newcommand{\m}{\mathbb{M}}

\newcommand{\C}{{\mathcal{C}}}

\newcommand{\bc}{{\mathbf{c}}}

\newcommand{\bzero}{{\mathbf{0}}}

\usepackage{blindtext}

\ifCLASSINFOpdf

\else

\fi

\begin{document}
%
\title{Two families of negacyclic BCH codes 
\thanks{
X. Wang's research was supported by The National Natural Science Foundation of China under Grant Number  12001175. 
Z. Sun's research was supported by The National Natural Science Foundation of China under Grant Number  62002093. 
C. Ding's research was supported by The Hong Kong Research Grants Council, Proj. No. $16301522$.}
}

\author{Xiaoqiang Wang\thanks{X. Wang is with The Hubei Key Laboratory of Applied Mathematics, Faculty of Mathematics and Statistics, Hubei University, Wuhan 430062, China (email: waxiqq@163.com).}, 
\and 
Zhonghua Sun\thanks{Z. Sun is with The School of Mathematics, Hefei University of Technology, Hefei, 230601, Anhui, China 
(email:  sunzhonghuas@163.com).}, \and 
Cunsheng Ding\thanks{C. Ding is with The Department of Computer Science and Engineering,
The Hong Kong University of Science and Technology, Clear Water Bay, Kowloon, Hong Kong, China (email: cding@ust.hk).}
}

\maketitle

\begin{abstract}

Negacyclic BCH codes are a subclass of neagcyclic codes and are the best linear codes in many cases. However, there have been very few results on negacyclic BCH codes. Let $q$ be an odd prime power and $m$ be a positive integer. The objective of this paper is to study
   negacyclic BCH codes with length $\frac{q^m-1}{2}$ and $\frac{q^m+1}{2}$ over the finite field $\gf(q)$ and analyse their parameters. The negacyclic BCH codes presented in this paper have good parameters in general, and contain many optimal linear codes. For certain $q$ and $m$, compared with cyclic codes with the same dimension and length, the negacyclic BCH codes presented in this paper have a larger minimum distance in some cases.

\end{abstract}

\begin{IEEEkeywords}
Constacyclic code, \and neagcyclic BCH code,  \and negacyclic code, \and linear code.
\end{IEEEkeywords}

%
\IEEEpeerreviewmaketitle

\section{Introduction}\label{sec-introduction}

Let $q$ be a prime power and $n$ be a positive integer. An $[n, k, d]$ linear code $\mathcal{C}$ over the finite field $\gf(q)$ is a $k$-dimensional linear subspace of $\gf(q)^n$ with minimum Hamming distance $d$. The dual code of $\mathcal{C}$ is defined by
$$\mathcal{C}^{\perp}=\{\mathbf{b} \in \gf(q)^n\,:\,\mathbf{b}\mathbf{c}^{T}=0 \,\,\text{for any $\mathbf{c} \in \mathcal{C}$}\}, $$
where $\mathbf{b}\mathbf{c}^{T}$ is the standard inner product of two vectors $\mathbf{b}$ and $\mathbf{c}$. An $[n,k]$ linear code $\C$ over $\gf(q)$ is said to be {\it negacyclic} if
$(c_0,c_1, c_2, \ldots, c_{n-1}) \in \C$ implies $(-c_{n-1}, c_0, c_1, c_2,\ldots, c_{n-2})
\in \C$.
 By identifying any vector $(c_0,c_1,\ldots,c_{n-1}) \in \gf(q)^n$ with the polynomial
$$c_0+c_1x+c_2x^2+\cdots+c_{n-1}x^{n-1} \in \gf(q)[x]/(x^n+1),$$
any negacyclic code $\C$ of length $n$ over $\gf(q)$ corresponds to an ideal of the quotient ring
$\gf(q)[x]/(x^n+1)$.  It is known that every ideal of $\gf(q)[x]/(x^n+1)$ must be principal. Thus, every negacyclic code $\C$ can be expressed as $\C=(g(x))$, where
$g(x)$ is a monic polynomial with the smallest degree and is called the {\it generator polynomial} of $\C$. Let $h(x)=(x^n+1)/g(x)$, then $h(x)$ is referred to as the {\em check polynomial} of $\C$.
The zeros of $g(x)$ and $h(x)$ are called zeros and non-zeros of $\C$~\cite{MacWilliams77}. The dual of $\C$ is also a negacyclic code and has generator polynomial $g^{*}(x)=h_0^{-1}x^{k}h(x^{-1})$, i.e., the \emph{reciprocal polynomial} of $h(x)$, where $h(x)=h_0+h_1x+ \cdots + h_{k-1}x^{k-1}+x^k$.

 An $[n, k, d]$ linear code over $\gf(q)$ is said to be \textit{distance-optimal} (respectively,
\textit{dimension-optimal} and \textit{length-optimal}) if there does not exist an $[n, k, d' \geq d+1]$ (respectively,
$[n, k' \geq k+1, d]$ and $[n' \leq n-1, k, d]$) linear code over $\gf(q)$. A code is said to be \emph{optimal} if it is
 length-optimal, or dimension-optimal, or distance-optimal, or meets a bound for linear codes.

In this paper, we always let $n$ be a positive integer and $q$ be an odd prime power with $\gcd(n,q)=1$. Similar to cyclic BCH codes, negacyclic BCH codes and their designed distances are defined as follows. Let $\alpha$ be a primitive element of $\gf(q^\ell)$, where $\ell=\ord_{2n}(q)$ is the order of $q$ modulo $2n$. Put $\beta=\alpha^{(q^\ell-1)/2n}$,
then $\beta$ is a primitive $2n$-th root of unity in $\gf(q^\ell)$.
 For any $i$ with $0\leq i\leq n-1$,
let $\m_{\beta^{1+2i}}(x)$ denote the minimal polynomial of $\beta^{1+2i}$ over $\gf(q)$. For any $\delta$
 with $2 \leq \delta \leq n$, let
\begin{eqnarray}\label{eqn-BCHdefiningSet}
g_{(q,n,\delta,b)}(x)=\lcm \left(\m_{\beta^{1+2b}}(x), \m_{\beta^{1+2(b+1)}}(x), \ldots, \m_{\beta^{1+2(b+\delta-2)}}(x)\right),
\end{eqnarray}
where $b$ is an integer and lcm denotes the least common multiple of these minimal polynomials. Let $\C_{(q,n,\delta,b)}$ denote
the negacyclic code of length $n$ over $\gf(q)$ with generator polynomial $g_{(q,n,\delta,b)}(x)$, then $\C_{(q,n,\delta,b)}$ is called a \emph{negacyclic BCH code} with {\it designed distance} $\delta$.

Cyclic BCH codes were introduced in 1959 by Hocquenghem \cite{Hocquenghem59}, and independently in 1960 by Bose and Ray-Chaudhuri \cite{Bose62}. They were extended to cyclic BCH
codes over finite fields by Gorenstein and Zierler in 1961 \cite{Gorenstein61}. In the past decade,  cyclic BCH codes have been widely studied and are treated in almost every book on coding theory as 
cyclic BCH codes are a special class of cyclic codes with interesting properties and applications, and are usually among the best cyclic codes.
The reader is referred to, for example, \cite{Aly07,Augot94,Charpin90,Ding2015,Ding15,Ding17,Lid017,Lid17,Liu17, Yue15,Dianwu96} for information on cyclic BCH codes. However, litter progress on the study of the dimension and minimum distance of negacyclic BCH codes has been made.

Negacyclic codes over finite fields  were initiated by Berkelamp in 1968 \cite{Berlekamp68,Berlekamp681}. Krishna and Sarwate \cite{KS90} found that negacyclic codes over finite fields can product optimal linear codes in many cases. Since then, a lot of quantum codes with good parameters have been constructed from negacyclic BCH codes \cite{Kai13,Kai131,Guo20,Zhu20}. Recently, Zhu et al. \cite{Zhu20} studied the negacyclic BCH codes of length $(q^{2m}-1)/(q-1)$ with designed distance $\delta \leq q^m+2$ and determined their dimensions. At the same time, Pang et al. \cite{Pang18} researched three classes of negacyclic BCH codes with designed distance in some ranges.
It is known that coset leaders provide information on the Bose distance and dimension of cyclic BCH codes. Similarly, odd coset leaders provide information on the minimum distance and dimension of negacyclic BCH codes. However, until now 
few results on the first several largest odd coset leaders are known.

With  the help of Magma, we found that for some codes with the same dimension and length, the best negacyclic BCH codes have better parameters than the best cyclic codes. This shows that negacyclic BCH codes have better parameters than cyclic BCH codes in some cases. Some examples of such code parameters are given in Table \ref{table1}.

\begin{table}[h]
{\caption{\rm The best cyclic codes and negacyclic BCH codes.
}\label{table1}
\begin{center}
\begin{tabular}{ccccc}\hline
$q$     &$m$                       & best cyclic codes  & negacyclic BCH codes\\\hline
3    &5                                              &[122,112,2]      &[122,112,5]    \\\hline
3    &3                                               &[14,8,2]      &[14,8,5]    \\\hline
3    &4                                                &[40,36,2]      &[40,36,3]    \\\hline
3    &4                                               &[40,28,5]      &[40,28,6]    \\\hline
3    &4                                           &[40,8,20]      &[40,8,21]    \\\hline
\end{tabular}
\end{center}}
\end{table}

Inspired and motivated by the examples of negacyclic codes in Table \ref{table1} and earlier works on negacyclic codes,  we study the negacyclic BCH codes of length $\frac{q^m-1}{2}$ and $\frac{q^m+1}{2}$ over $\gf(q)$ in this paper.
The first objective of this paper is to analyse the parameters
 of negacyclic BCH codes of length $\frac{q^m-1}{2}$ and $\frac{q^m+1}{2}$ with small  and large dimensions. The second objective
 of this paper is to
 determine the exact Hamming distance of neagcyclic BCH codes  of  length $\frac{q^m-1}{2}$ and $\frac{q^m+1}{2}$ with generator polynomials being an irreducible polynomial or the product of two irreducible polynomials.
 To investigate the optimality
 of the codes studied in this paper, we compare them with the tables of the best known linear codes maintained in \cite{Grassl2006}, and show that some of  the codes presented in this paper are optimal.

The rest of this paper is organized as follows. Section~\ref{sec-auxiliary} introduces some preliminaries. Section~\ref{sec-MDSqeven}
and Section~\ref{sec-MDSqodd}
study the parameters
 of negacyclic BCH codes of length $\frac{q^m-1}{2}$ and $\frac{q^m+1}{2}$ with small  and large dimensions, respectively. Section~\ref{sec-largedimes}
determines the parameters of BCH codes with generator polynomials being an irreducible polynomial or the product of two irreducible polynomials. Section \ref{sec-finals} concludes this paper.

\section{Preliminaries}\label{sec-auxiliary}

In this section, we introduce some
basic concepts and known results, which will be used later in this paper.
\subsection{Notation used starting from now on}
Starting from now on, we adopt the following notation unless otherwise stated:
\begin{itemize}
\item $\gf(q)$ is the finite field with $q$ elements.
\item $\C_{(q,n,\delta,b)}$ denotes 
the negacyclic BCH code of length $n$ over $\gf(q)$ with generator polynomial $g_{(q,n,\delta,b)}(x)$, where $g_{(q,n,\delta,b)}(x)$ was defined in (\ref{eqn-BCHdefiningSet}).
\item $[u,v]:=\{u,u+1,\ldots,v\},$ where $u,v$ are nonnegative integers with $u \leq v$. If $u > v$, then $[u, v]=\emptyset$. 
\item $\delta_i$ is the $i$-th largest odd coset leader modulo $q^m-1$ and $q^m+1$ in Section III and Section IV, respectively.
\item $\lfloor x \rfloor$ denotes the largest integer less than or equal to $x$.
\item $\lceil x \rceil$ denotes the smallest integer greater than or equal to $x$. 
\item $u \bmod v$ denotes the nonnegative remainder in $[0, v-1]$ when $u$ is divided by $v$, where $u$ and $v \geq 2$ are 
integers.    
\end{itemize}

\subsection{Negacyclic LCD  codes}
A linear code $\C$ is called an {\it LCD code} ({\it linear complementary dual}) if $\C\cap \C^{\bot}=\{\bzero\}$. Let $$g(x)=g_tx^t+g_{t-1}x^{t-1}+\cdots+g_1x+g_0$$ be a polynomial over $\gf(q)$ with $g_t\neq 0$ and $g_0\neq 0$, the reciprocal $g^*(x)$ of $g(x)$ is defined by
$$g^*(x)=g_0^{-1}x^tg(x^{-1}).$$
The following basic results about negacyclic LCD codes over finite fields were given in \cite{Pang18}.
\begin{lemma}\cite[Theorem 3.2]{Pang18}\label{lemma1}
Let $\mathcal{C}$ be a negacyclic code of length $n$ over $\gf(q)$ with generator polynomial $g(x)$. Then the following statements are equivalent.
\begin{itemize}
\item $\mathcal{C}$ is an LCD code.
\item $g(x)$ is self-reciprocal.
\item $\gamma^{-1}$ is a root of $g(x)$ for every root $\gamma$ of $g(x)$ over the splitting field of $g(x)$.
\end{itemize}
Furthermore, the negacyclic code of length $n$ over $\gf(q)$ is LCD if $-1$ is a power of $q$ modulo $2n$.
\end{lemma}

By Lemma \ref{lemma1}, it is straightforward that every negacyclic code of length $\frac{q^m+1}{2}$ is LCD. However, this is not true for negacyclic codes with length $\frac{q^m-1}{2}$. In \cite{Pang18}, the authors determined the number of LCD negacyclic codes with length $\frac{q^m-1}{2}$.

\subsection{Some bounds of linear codes}

Similar to cyclic codes, negacyclic codes over finite fields have the following BCH bound.

\begin{lemma}\label{lem:2}\cite[Lemma 4]{KS90}
Let $\mathcal{C}$ be a negacyclic code of length $n$ over $\gf(q)$ and $\gamma \in  \gf(q^m)$ be a primitive 2n-th root of unity. Let $g(x)$ be the generator polynomial of $\mathcal{C}$. If there are integers $e$, $h$, $\delta$ with $\gcd(e,n)=1$ and $2\leq \delta \leq n$ such that
$$g(\gamma^{1+2eh})=g(\gamma^{1+2e(h+1)})=\cdots=g(\gamma^{1+2e(h+\delta-2)})=0,$$
then the minimum Hamming distance of $\mathcal{C}$ is at least $\delta$.
\end{lemma}

For the negacyclic BCH codes with length $\frac{q^m+1}{2}$, we have the following bound, which is much better than the bound of Lemma \ref{lem:2} when $\delta$ is getting large.

\begin{lemma}\label{lem:3}\cite[Theorem 4.3]{Pang18}
Let $n=\frac{q^m+1}{2}$,
 then the code $\mathcal{C}_{(q,n,\delta,0)}$ has minimum distance $d\geq 2\delta-1$.
\end{lemma}

The following is a well-known bound for any code.

\begin{lemma}\cite[Lemma 6]{Rouayheb2007}\label{bound2}
Let $q$ be an odd prime power and $A_q(n,d)$ be the maximum number of codewords of a  code with length $n$ over $\gf(q)$ and minimum Hamming distance at least $d$. Let $q\geq 3$, $t=n-d+1$ and $r=\lfloor \min\{\frac{n-t}{2}, \frac{t-1}{q-2}\}\rfloor$. Then
\begin{equation*}
\begin{split}
A_q(n,d)\leq \frac{q^{t+2r}}{\sum_{i=0}^r\left(\begin{array}{cccc}
   t+2r  \\
     i  \\
\end{array}
\right)(q-1)^i}.
\end{split}
\end{equation*}
\end{lemma}

\subsection{Some known and basic results}

Let $\mathbb{Z}_N$ denote the ring of integers modulo $N$, which is a positive integer. Let $i$ be an integer with $0\leq i \leq N-1$. The {\it $q$-cyclotomic coset} of $i$ modulo $N$ is defined by
$$C_i^{(q,N)}=\{i, iq, iq^2, \ldots, iq^{\ell_{i-1}}\} \bmod N \subseteq \mathbb{Z}_N, $$
where $\ell_i$ is the smallest positive integer such that $i\equiv iq^{\ell_i} \bmod N$, and is the size of the $q$-cyclotomic coset. The smallest integer in $C_i^{(q, N)}$ is called the {\it coset leader} of $C_i^{(q, N)}$.

\begin{lemma}\label{lem:7}\cite{Liuv2, Yan2018}
Let $q$ be an odd prime and $m\geq 2$. Let $\phi_1$, $\phi_2$ and $\phi_3$ be the first three largest coset leaders modulo $q^m+1$, respectively. Then  $\phi_1=\frac{q^m+1}{2}$ for any $m$, and $\phi_2=\frac{(q-1)(q^m+1)}{2(q+1)}$ and $\phi_3=\frac{(q-1)(q^m-2q^{m-2}-1)}{2(q+1)}$ for odd $m$. Moreover, $|C_{\phi_1}^{(q,2n)}|=1$ for any $m$, and $|C_{\phi_2}^{(q,2n)}|=2$ and $|C_{\phi_3}^{(q,2n)}|=2m$ for odd $m$.
\end{lemma}

\begin{lemma}\label{lem:8}\cite[Theorem 13]{Ding2015}
Let $q$ be an odd prime and $m$ be an integer. Then the largest coset leader modulo $q^m-1$ is $(q-1)q^{m-1}-1$. Moreover, the size of the coset containing this largest coset leader is $m$.
\end{lemma}

The following lemma documents all the $q$-cyclotomic coset leaders modulo $q-1$ and $q+1$.

\begin{lemma}\label{lem:9}
Let $q$ be an odd prime. If $n=q+1$, then a nonnegative integer $i \leq q+1$ modulo $n$ is a coset leader if and only if $0\leq i\leq \frac{q+1}{2}$.
If $n=q-1$, then every nonnegative integer $i \leq q-2$ modulo $n$ is a coset leader.
\end{lemma}

\begin{proof}
We prove the conclusions of this lemma only for the case that $n=q+1$, and omit the proof of the conclusions for $n=q-1$, which are  obvious.

It is easy to see that the $q$-cyclotomic cosets modulo $q+1$
are the following
$$
\{0\} \mbox{ and } \{i, q+1-i\}, \ 1 \leq i \leq  \frac{q+1}{2}.
$$
By the definition of the coset leaders, we have the desired result.
\end{proof}

Combining Lemmas \ref{lem:2}, \ref{lem:3}, \ref{lem:9} and the sphere packing bound, we obtain the following results.

\begin{proposition}
Let $n=\frac{q-1}{2}$ and $2\leq \delta \leq \frac{q-1}{2} $ be an integer.
Then the negacyclic BCH code $\C_{(q,n,\delta,0)}$ has parameters $\left[n,n-(\delta-1),\delta\right].$
 All of these negacyclic BCH codes are MDS.
\end{proposition}

\begin{proposition}
 Let $n=\frac{q+1}{2}$ and $2\leq \delta \leq \left\lfloor\frac{q+3}{4} \right\rfloor$ be an integer.
Then the negacyclic BCH code $\C_{(q,n,\delta,0)}$ has parameters $\left[n, n-2\left(\delta-1\right) , 2\delta-1 \right].$
 All of these negacyclic BCH codes are MDS.
\end{proposition}

Let $c(x)$ be a polynomial in $\gf(q)[x]/( x^n+1)$. Define an isomorphic map
\begin{equation*}
\begin{split}
\varphi\,\,:\,\, \gf(q)[x]/( x^n&+1)\,\,\to \,\, \gf(q)[x]/(x^n-1)\\
&c(x)\,\,\,\,\mapsto \,\,c(-x).
\end{split}
\end{equation*}
Then we have the following results.

\begin{lemma}\label{lem:12}
Let $\varphi$ be defined as above. Let $n$ be odd and $\C$ be a negacyclic code of length $n$ over $\gf(q)$, then $\C$ and $ \varphi(\C)$ have the same parameters.
\end{lemma} 

\subsection{The $q$-weight of integers}
Let $s$ be an integer and the $q$-adic expansion of $s$ be $s=\sum_{i=0}^{m-1}s_iq^i$, where $0\leq s_i\leq q-1$. Define the {\it $q$-weight} of $s$ to be $w_q(s)=\sum_{i=0}^{m-1}s_i$ and the {\it sequence} of $s$ to be
$$\overline{s}=(s_{m-1},s_{m-2},\ldots,s_1,s_0).$$ 
For two positive integers $A$ and $B$ with $\overline{A}=(a_{m-1},a_{m-2},\ldots,a_{0})$ and $\overline{B}=(b_{m-1},b_{m-2},\ldots,b_{0})$, we write $\overline{A}>\overline{B}$ and say that $\overline{A}$ is greater than $\overline{B}$  if there exists an integer $0\leq i\leq m-1$ such that $a_i> b_i$ and $a_j=b_j$ for $j\in [i+1,m-1]$, and write $\overline{A}=\overline{B}$ and 
say that $\overline{A}$ equals $\overline{B}$  if $a_i=b_i$ for all $i \in [0,m-1]$. It is clear that $A>B$ if and only if $\overline{A}>\overline{B}$ and $A=B$ if and only if  $\overline{A}=\overline{B}$. 
Let $1\leq j\leq m$, it is easily seen that the sequence of $q^js \bmod {(q^m-1)}$ is
$$\overline{q^js \bmod {(q^m-1)}}=(s_{m-j-1},s_{m-j-2},\ldots,s_{m-j+1},s_{m-j}),$$
which is called the \emph{circular $j$-left-shift} of $(s_{m-1}, \ldots, s_0)$,  
where the subscript of each coordinate is regarded as an integer modulo $m$. With the preparations above, we have the following result, which is straightforward by definition.

\begin{lemma}\label{lem:5}
Let $0<i<q^m-1$ and $\overline{i}=(i_{m-1},i_{m-2},\ldots, i_0)$. Then $i$ is a coset leader modulo $q^m-1$ if and only if the circular 
$j$-left-shift of $(i_{m-1},i_{m-2}, \ldots,  i_0)$ is grater than or equal to $(i_{m-1},i_{m-2}, \ldots, i_0)$ for each $0 \leq j \leq m-1$.
\end{lemma}

\section{Negacyclic BCH codes with length $\frac{q^m-1}{2}$}\label{sec-MDSqeven}

In this section, we always let $m\geq 2$, $n=\frac{q^m-1}{2}$ and study the negacyclic BCH codes over $\gf(q)$ with length $n$. Throughout this section, whenever we say ``$x$ is a coset leader", we mean that ``$x$ is a coset leader modulo $q^m-1$". That is to say,  we omit the phrase ``modulo $q^m-1$". We start with the following lemmas, which will be useful for us to calculate the parameters of this family of negacyclic BCH codes with large dimensions.

\begin{lemma}\label{lem:13}\cite{Liu17}
 Let $m\geq 3$ be odd and $1\leq i\leq q^{(m+1)/2}-1$. Then $i$ is a coset leader if and only if $i\not\equiv 0 \pmod {q}$. Moreover, $|C_i^{(q,2n)}|=m$.
\end{lemma}

\begin{lemma}\label{lem:14}\cite{Liu17}
 Let $m\geq 2$ be even and $1\leq i\leq 2q^{m/2}-1$. Then $i$ is a coset leader if and only if $i\not\equiv 0 \pmod{q}$. Moreover, $|C_{q^{m/2}+1}^{(q,2n)}|=\frac{m}{2}$ and $|C_i^{(q,2n)}|=m$ if $i\neq q^{m/2}+1$ and $i$ is a coset leader.
\end{lemma}

\begin{theorem}\label{thm:15}
Let $2\leq \delta \leq q^{m/2}+1$ if $m$ is even and $2\leq \delta \leq \frac{q^{(m+1)/2}+1}{2}$ if $m$ is odd. Then the negacyclic BCH code $\C_{(q,n,\delta,0)}$ has parameters
$$\left[n, \, n-m\left \lceil \frac{(2\delta-3)(q-1)}{2q} \right\rceil, \,  d \right]$$
and generator polynomial
\begin{align*}
	g(x)=\prod_{\substack{0\leq i\leq \delta-2\\ i\not\equiv \frac{q-1}{2} \pmod{q}}} \m_{\beta^{1+2i}}(x),
\end{align*}
where 
\begin{align*}
	d \geq \begin{cases}
	\delta +1 &{\rm if~}\delta \equiv \frac{q+1}2 \pmod{q},  \\
	\delta &{\rm otherwise}.
\end{cases}
\end{align*}
\end{theorem}

\begin{proof} 
If $i$ satisfies $0\leq i\leq \delta-2$ and $1+2i \not\equiv 0 \pmod{q}$, then it follows from Lemmas \ref{lem:13} and \ref{lem:14} 
that  $1+2i$ is a coset leader. By the definition of negacyclic BCH codes, we have the desired generator polynomial of $\C_{(q, n, \delta, 0)}$.

It is easily seen that $1+2i \equiv 0 \pmod{q}$ if and only if $i\equiv \frac{q-1}2 \pmod{q}$. Consequently, the number of integers 
$i$ satisfying $0\leq i\leq \delta-2$ and $1+2i\equiv 0 \pmod{q}$ is 
$$\left\lfloor \frac{2\delta-3-q}{2q} \right\rfloor +1.$$
Hence, the number of integers $i$ satisfying $0\leq i\leq \delta-2$ and $1+2i\not\equiv 0 \pmod{q}$ is 
$$\delta-2-\left\lfloor \frac{2\delta-3-q}{2q} \right\rfloor=\left \lceil \delta-2- \frac{2\delta-3-q}{2q} \right\rceil=\left \lceil \frac{(2\delta-3)(q-1)}{2q} \right\rceil.$$
From Lemmas \ref{lem:13} and \ref{lem:14}, we know that $|C_{1+2i }^{(q,2n)}|=m$. Thus, the dimension of $\C_{(q,n,\delta,0)}$ is
$$n-m\left \lceil \frac{(2\delta-3)(q-1)}{2q} \right\rceil.$$
Note that $1+2(\delta-1)$ is not a coset leader if $\delta \equiv \frac{q+1}2 \pmod{q}$. By Lemma \ref{lem:2}, we have the desired lower bounds on the minimum distance of $\C_{(q,n,\delta,0)}$. This completes the proof.
\end{proof}

\begin{example}\label{example-01}
We have the following examples of the negacyclic BCH code of Theorem \ref{thm:15}.
\begin{itemize}
\item Let $q=3$ and $m=5$, then the code $\mathcal{C}_{(3,121,6,0)}$ has parameters $[121,106, 6]$.
\item Let $q=3$ and $m=4$, then the code $\mathcal{C}_{(3,40,6,0)}$ has parameters $[40,28,6]$.
\item Let $q=5$ and $m=2$, then the code $\mathcal{C}_{(5,12,4,0)}$ has parameters $[12,8, 4]$.
\item Let $q=7$ and $m=2$, then the code $\mathcal{C}_{(7,24,5,0)}$ has parameters $[24,18, 5]$.
\item Let $q=9$ and $m=2$, then the code $\mathcal{C}_{(9,40,6,0)}$ has parameters $[40,32, 6]$.
\end{itemize}
All the five codes are distance-optimal according to the tables of best codes known  in \cite{Grassl2006}, 
and their minimum distances achieve the lower bounds in Theorem \ref{thm:15}. 
\end{example}

In the following, we study the negacyclic BCH codes of length $\frac{q^m-1}{2}$ with small dimensions. To this end, we need to find the first few largest odd coset leaders modulo $q^m-1$. From Lemma \ref{lem:8}, we know that the first largest coset leader is $(q-1)q^{m-1}-1$. Then we have the following result.

\begin{lemma}\label{lem:17}
Let $q$ be an odd prime power, then the first largest odd coset leader is given by
$$\delta_1=(q-1)q^{m-1}-1.$$
\end{lemma}

To the best of our knowledge, except for the first largest coset leader modulo $q^m-1$ which is odd, the known first few largest coset leaders modulo $q^m-1$ are even. Hence, we will find out the second largest odd coset leader $\delta_2$ and the third largest odd coset leader $\delta_3$. We start with the following lemma.

\begin{lemma}\label{lem:18}
Let $q$ be an odd prime power, then there always exists a coset leader $S_i$ satisfying $w_q(S_i)=(q-1)m-i$ for $1\leq i\leq m$, where $w_q(S_i)$ is the $q$-weight of $S_i$. Let $M_i$ and $M_{i+1}$ be the largest coset leaders with $w_q(M_i)=(q-1)m-i$ and $w_q(M_{i+1})=(q-1)m-(i+1)$, respectively. Then $M_i>M_{i+1}$ for $1\leq i\leq m-1$.
\end{lemma}
\begin{proof}
Let $\overline{S}_i=(\underbrace{q-2, \ldots, q-2,}_{i} \underbrace{q-1, \ldots, q-1}_{m-i}).$ By definition, it is easy to see that $S_i$ is a coset leader with $w_q(S_i)=(q-1)m-i$. Hence, there always exists a coset leader $S_i$ satisfying $w_q(S_i)=(q-1)m-i$ for $1\leq i\leq m$.

 We now prove that  for $1\leq i\leq m-1$, we have $M_i>M_{i+1}$ if $M_i$ and $M_{i+1}$ are the largest coset leaders with $w_q(M_i)=(q-1)m-i$ and $w_q(M_{i+1})=(q-1)m-(i+1)$, respectively.

Let the sequence of $M_i$ be 
$$\overline{M}_i=(b_{m-1},b_{m-2},\ldots,b_1,b_0).$$ 
If there exists a positive integer $0\leq i\leq m-1$ such that $b_i<q-2$, then the sequence of $M_iq^{m-1-i} \bmod{(q^m-1)}$ is 
\begin{equation*}
\begin{split}
(b_i,b_{i-1},\ldots,b_0,b_{m-1},b_{m-2},\ldots,b_{i+1}).
\end{split}
\end{equation*}
It is easily seen that $S_i$ is larger than $M_iq^{m-1-i} \bmod{(q^m-1)}$, which is contradictory to that $M_i$ is the largest coset modulo $q^m-1$ with $w_q(S_i)=(q-1)m-i$. Hence, $b_{m-1},b_{m-2},\ldots,b_1,b_0 \in \{q-1,q-2\}$.

If $b_{m-1}=q-1$, then the sequence of $M_i$ must be
$$\overline{M_i}=(q-1,q-1,\ldots,q-1 ).
$$
Otherwise, there is a $b_i=q-2$, then $M_iq^{m-1-i} \bmod{(q^m-1)}$ $<$ $M_i$, which is contradictory to that $M_i$ is a coset leader modulo $q^m-1$. Hence, the sequence of $M_i$ can be expressed as
\begin{equation}
\overline{M_i}=(\mathop{q-2}\limits_{m-1},\underbrace{q-1, \ldots, q-1,}_{m-2-l_1} \mathop{q-2}\limits_{l_1},\ldots,\mathop{q-2}_{l_{i-2}},
\underbrace{q-1, \ldots, q-1}_{ l_{i-2}-l_{i-1}-1},\mathop{q-2}_{l_{i-1}},
\underbrace{q-1, \ldots, q-1}_{l_{i-1}}),
\end{equation}
where $l_1,l_2,\ldots,l_{i-1}$ are positive integers.

Similarly, since $M_{i+1}$ is the largest coset leader with $w_q(M_{i+1})=(m-1)q-(i+1)$, the sequence of $M_{i+1}$ can be expressed as
\begin{equation}\label{eq:12-09-04}
\overline{M_{i+1}}=(\mathop{q-2}\limits_{m-1},\underbrace{q-1, \ldots, q-1,}_{m-2-j_1} \mathop{q-2}\limits_{j_1},\ldots,\mathop{q-2}_{j_{i-1}},
\underbrace{q-1, \ldots, q-1}_{ j_{i-1}-j_i-1},\mathop{q-2}_{j_i},
\underbrace{q-1, \ldots, q-1}_{j_i}),
\end{equation}
where $j_1,j_2,\ldots,j_{i}$ are positive integers.
Let $\overline{M'_i}$ be the sequence obtained by replacing the positive integer $q-2$ at the coordinate $j_i$ with $q-1$ in sequence $\overline{M_{i+1}}$, i.e., 
\begin{equation}
\overline{M'_i}=(\mathop{q-2}\limits_{m-1},\underbrace{q-1, \ldots, q-1,}_{m-2-j_1} \mathop{q-2}\limits_{j_1},\ldots,\mathop{q-2}_{j_{i-2}},
\underbrace{q-1, \ldots, q-1}_{ j_{i-2}-j_{i-1}-1},\mathop{q-2}_{j_{i-1}},
\underbrace{q-1, \ldots, q-1}_{j_{i-1}}).
\end{equation}

We now prove that $M'_i$ is a coset leader.
If $M'_i$ is not a coset leader,  then there exists a positive integer $t$ such that
\begin{equation}\label{eq:1210-01}
 M'_iq^{t} \bmod{(q^m-1)} <M'_i.
\end{equation}
It is obvious that these $t=m-1-j_h$ with $1\leq h\leq i-1$ are the possible values such that (\ref{eq:1210-01}) holds. Otherwise,
 $$\overline{M'_iq^{t} \bmod{(q^m-1)}}=(q-1,\ldots )>\overline{M'_i},$$
which is contradictory to (\ref{eq:1210-01}) since  $M'_iq^{t} \bmod {(q^m-1)}>M'_i$ if $\overline{M'_iq^{t} \bmod{(q^m-1)}}>\overline{M'_i}$.

It is obvious that $M_{i+1}=(q-1)q^{m-1}-q^{j_1}-q^{j_2}-\cdots-q^{j_i}-1$ and 
$$M'_i=(q-1)q^{m-1}-q^{j_1}-q^{j_2}-\cdots-q^{j_{i-1}}-1.$$ 
If $t=m-1-j_h$ for $1\leq h\leq i-1$, then
\begin{equation*}
\begin{split}
 M'_iq^{t} \bmod{(q^m-1)} &\equiv (M_{i+1}+q^{j_i})q^{t} \bmod{(q^m-1)} \\
 &\geq M_{i+1}+q^{j_i+t}\\
 &>M'_i,
\end{split}
\end{equation*}
which is contradictory to (\ref{eq:1210-01}).
Hence, $M'_{i}$ is a coset leader with $w_q(M'_{i})=(m-1)q-i$. Therefore, $M_i>M'_{i}>M_{i+1}$ since $M_i$ is the largest coset leader with $w_q(M_i)=(m-1)q-i$.
The desired conclusions then follow.
\end{proof}

With the help of Lemma \ref{lem:18}, we now determine the values of $\delta_2$ and $\delta_3$.

\begin{lemma}\label{lem:19}
Let $q\geq 3$ be an odd prime power.  Then the second largest odd coset leader $\delta_2=(q-1)q^{m-1}-q^{\lfloor\frac{2m-1}{3}\rfloor}-q^{\lfloor\frac{m-1}{3}\rfloor}-1$.
\end{lemma}

\begin{proof}
If $\delta_2$ is a coset leader with $q$-weight $m(q-1)-1$, it is easily seen that the sequence of $\delta_2$ must be
$$\overline{\delta_2}=(\mathop{q-2}\limits_{m-1},\underbrace{q-1, \ldots, q-1}_{m-1}).$$
Then $\delta_2=(q-1)q^{m-1}-1$, which is contradictory to that $(q-1)q^{m-1}-1$ is the largest odd coset leader 
(see Lemma \ref{lem:17}). 

If $\delta_2$ is a coset leader with $q$-weight $m(q-1)-2$, then the sequence of $\delta_2$ has the form
$$\overline{\delta_2}=(\mathop{q-2}\limits_{m-1},\underbrace{q-1, \ldots, q-1}_{m-2-i},\mathop{q-2}\limits_i,\underbrace{q-1, \ldots, q-1}_i)$$
for $0\leq i\leq m-1$. Hence, $\delta_2=(q-1)q^{m-1}-q^i-1$. However, $(q-1)q^{m-1}-q^i-1$ is even for any $i$. This is a contradiction.

Hence, the $q$-weight of $\delta_2$ is less than or equal to $m(q-1)-3$. If $m=2$, it is easily seen that $q^2-2q-2$ is the largest coset leader with the $q$-weight $2(q-1)-3$ since the sequence of $q^2-2q-2$ has the form $(q-3,q-2)$. Then from Lemma \ref{lem:18}, we have $\delta_2=q^2-2q-2$ if $m=2$.

In the following, we consider the case $m\geq 3$.
Let $M_3$ be the largest coset leader with $w_q(M_3)=m(q-1)-3$. Then the sequence of $M_3$ has the form
 \begin{equation}\label{eq:3}
\overline{M}_3=(\mathop{q-2}\limits_{m-1},\underbrace{q-1, \ldots, q-1,}_{m-2-i} \mathop{q-2}_{i}, \underbrace{q-1, \ldots, q-1,}_{i-j-1} \mathop{q-2}\limits_{j},
\underbrace{q-1, \ldots, q-1}_{j}),
\end{equation}
where $m-2-i\leq i-j-1 \leq j$. Since $M_3$ is the largest coset leader with $w_q(M_3)=m(q-1)-3$, from Lemma \ref {lem:5} we have that
\begin{itemize}
\item $m-2-i=\left\lfloor \frac{m-3}{3} \right\rfloor$, \,\, $j\geq \left\lfloor \frac{m-3}{3} \right\rfloor$ and $i-j-1\geq \left\lfloor \frac{m-3}{3} \right\rfloor$.
 \item If $\frac{m-3}{3}$ is an integer, then $m-2-i=j$. Otherwise, $m-1-i=j$.
\end{itemize}
Hence, we obtain
\begin{equation*}
\left\{
\begin{array}{ll}
i= \frac{2m}{3}-1\,\,\text{ and}\,\, j=\frac{m}{3}-1  & \mbox{ if $m\equiv 0 \pmod{3} $,}\\
i=m-1-\left\lfloor \frac{m}{3} \right\rfloor\,\,\text{ and}\,\, j=\left\lfloor \frac{m}{3} \right\rfloor  & \mbox{ if $m\equiv 1 \pmod{3}$,}\\
i=m-1-\left\lfloor \frac{m}{3} \right\rfloor\,\,\text{ and}\,\, j=\left\lfloor \frac{m}{3} \right\rfloor  & \mbox{ if $m\equiv 2 \pmod{3}$.}\\
\end{array}
\right.
\end{equation*}
Then from (\ref{eq:3}) and Lemma \ref{lem:18}, we have $\delta_2=(q-1)q^{m-1}-q^{\lfloor\frac{2m-1}{3}\rfloor}-q^{\lfloor\frac{m-1}{3}\rfloor}-1$ if $m\geq 3$.

Combining the discussions in the case $m=2$ and the case $m \geq 3$, we obtain the desired results.
\end{proof}

\begin{lemma}\label{lem:20}
Let $q\geq 3$ be an odd prime power and $\delta_2$ be defined in Lemma \ref{lem:19}, then
\begin{equation*}
|C_{\delta_2}^{(q,2n)}|=\begin{cases}
	m &{\rm if}~3\nmid m,\\
	\frac{m}3&{\rm if}~3\mid m.
\end{cases}
\end{equation*}
\end{lemma}

\begin{proof}
It is obvious that $|C_{\delta_2}^{(q,2n)}|$ is a divisor of $m$ since ord$_n(q)=m$. Let $|C_{\delta_2}^{(q,2n)}|=r$, then
\begin{equation*}
\delta_2(q^r-1)\equiv 0 \pmod{q^m-1}. 
\end{equation*}
The congruence equation above is equivalent to
\begin{equation}\label{eq:nqmodd1}
q^{\lceil \frac{2(m-1)}3 \rceil}+q^{\lfloor\frac{m+1}{3}\rfloor}+1 \equiv 0 \pmod{\frac{q^m-1}{q^r-1}},
\end{equation}
which implies that $m-r\leq \lceil \frac{2(m-1)}3 \rceil$, i.e., 
$$r\geq m- \left \lceil \frac{2(m-1)}3 \right \rceil= \left \lfloor \frac{m+2}3 \right \rfloor.$$
Note that $r\mid m$, we have $r\in \{m, \frac{m}2, \frac{m}3 \}$. If $r=\frac{m}2$ and $m$ is even, from (\ref{eq:nqmodd1}) we have  
$$q^{\lfloor\frac{m+1}{3}\rfloor}-q^{\lceil \frac{m-4}6 \rceil}+1 \equiv 0 \pmod{q^{\frac{m}2}+1}.$$
When $m=2$, $q \equiv 0 \pmod{q+1}$, which is impossible. When $m\geq 4$, $\lfloor\frac{m+1}{3}\rfloor< \frac{m}2$, which is  impossible. If $r=\frac{m}3$ and $3 \mid m$, it is easy to check that (\ref{eq:nqmodd1}) holds. This means that $|C_{\delta_2}^{(q,2n)}|=\frac{m}3$. The desired conclusion then follows.
\end{proof}

\begin{lemma}\label{lem:21}
Let $q\geq 3$ be an odd prime power and let $m$ be a positive integer such that $q^m\geq 25$. Then
\begin{equation*}
\delta_3=\begin{cases}
	(q-1)q^{m-1}-q^{\lceil\frac{2m-1}{3}\rceil}-q^{\lfloor\frac{m}{3}-1\rfloor}-1 &{\rm if}~3\nmid(m+1),\\
	(q-1)q^{m-1}-q^{\frac{2m-1}{3}}-q^{\frac{m+1}{3}}-1 &{\rm if}~3\mid (m+1).
\end{cases}
\end{equation*}
Moreover, $|C_{\delta_3}^{(q,2n)}|=m$.
\end{lemma}
\begin{proof}
 We prove the conclusions of this lemma only for the case that $3 \nmid (m+1)$. The proof of the conclusions for $3\mid (m+1)$ is 
 similar and omitted.

Let
\begin{equation}\label{M:val}
M=\left[(q-1)q^{m-1}-q^{\lceil\frac{2m-1}{3}\rceil}-q^{\lfloor\frac{m}{3}-1\rfloor}-1\right]q^i \bmod{(q^m-1)}, 
\end{equation} 
where $1 \leq i \leq m-1$. 
To show that $(q-1)q^{m-1}-q^{\lceil\frac{2m-1}{3}\rceil}-q^{\lfloor\frac{m}{3}-1\rfloor}-1$ is a coset leader, 
we only need to prove that for any $1\leq i\leq m-1$, 
\begin{equation}\label{eq:1fdfM}
M\geq (q-1)q^{m-1}-q^{\lceil\frac{2m-1}{3}\rceil}-q^{\lfloor\frac{m}{3}-1\rfloor}-1.
\end{equation}
The proof will be carried out by distinguishing the following three cases.

\noindent{\bf Case 1}: $1\leq i\leq \lfloor\frac{m-2}3\rfloor$. Note that $\lceil\frac{2m-1}3 \rceil+\lfloor \frac{m-2}3 \rfloor=m-1$. From (\ref{M:val}), we have
\begin{equation}\label{val:M}
\begin{split}
M=q^m-q^{\lceil \frac{2m-1}{3}\rceil+i}-q^{\lfloor \frac{m}{3}-1\rfloor+i}-q^{i-1}-1.
\end{split}
\end{equation}
If $i<\lfloor\frac{m-2}3\rfloor$, it is easily seen that (\ref{eq:1fdfM}) holds.
If $i=\lfloor\frac{m-2}3\rfloor$, then (\ref{val:M}) becomes
$$M=q^m-q^{m-1}-q^{\lfloor \frac{2m-6}{3}\rfloor}-q^{\lfloor\frac{m-5}{3}\rfloor}-1.$$
It is easy to check that $\lfloor \frac{2m-6}{3}\rfloor < \lceil \frac{2m-1}{3}\rceil$ for any $m\geq 3$, then (\ref{eq:1fdfM}) holds and the equality never holds.

\noindent{\bf Case 2}: $\lfloor\frac{m-2}3\rfloor+1\leq i\leq  \lceil \frac{2m}3 \rceil$. From (\ref{M:val}) we have
\begin{equation*}
\begin{split}
M=q^m-q^{\lfloor \frac{m}{3}-1\rfloor+i}-q^{i-1}-q^{i-\lfloor \frac{m-2}3 \rfloor-1}-1.
\end{split}
\end{equation*}
With an analysis similar as in Case 1, we can prove that (\ref{eq:1fdfM}) holds and the equality never holds.

\noindent{\bf Case 3}: $\lceil\frac{2m}3 \rceil+1\leq i\leq m-1$. From (\ref{M:val}) we have
\begin{equation*}
\begin{split}
M=q^m-q^{i-1}-q^{i-\lfloor \frac{m-2}3 \rfloor-1}-q^{i-\lceil \frac{2m}{3}\rceil-1} -1.
\end{split}
\end{equation*}
It is obvious that $i-1<m-1$ as $i\leq m-1$. Then (\ref{eq:1fdfM}) holds and the equality never holds.

Collecting the conclusions in all the cases above,   we deduce that $(q-1)q^{m-1}-q^{\lceil\frac{2m-1}{3}\rceil}-q^{\lfloor\frac{m}{3}-1\rfloor}-1$ is a coset leader if $3 \nmid (m+1)$. Since the equality never holds for any $1\leq i\leq m-1$, we have $|C_{\delta_3}^{(q,2n)}|=m$.

 We now prove that there does not exist an odd coset leader in the range $[\delta_3-1, \delta_2-1]$. Since $3 \nmid (m+1)$, we have $m \equiv 0 \pmod{3}$ or $m \equiv 1 \pmod{3}$. When $m \geq 6$, it is easily seen that the sequence of $\delta_3$ is
\begin{equation*}
\begin{split}
\overline{\delta}_3=(\mathop{q-2}_{m-1},\underbrace{q-1, \ldots, q-1,}_{\frac{m}{3}-2} \mathop{q-2}_{\frac{2m}{3}},
\underbrace{q-1, \ldots, q-1,}_{ \frac{m}{3}} \mathop{q-2}_{\frac{m}{3}-1},
\underbrace{q-1, \ldots, q-1}_{ \frac{m}{3}-1})
\end{split}
\end{equation*}
if $m \equiv 0 \pmod{3}$ and
\begin{equation*}
\begin{split}
\overline{\delta}_3=(\mathop{q-2}\limits_{m-1},\underbrace{q-1, \ldots, q-1,}_{ \frac{m-7}{3}} \mathop{q-2}_{\frac{2m+1}{3}},
\underbrace{q-1, \ldots, q-1,}_{ \frac{m+2}{3}} \mathop{q-2}_{\frac{m-4}{3}},
\underbrace{q-1, \ldots, q-1}_{\frac{m-4}{3}})
\end{split}
\end{equation*}
if $m \equiv 1 \pmod{3}$.

Since $\delta_2=(q-1)q^{m-1}-q^{\lfloor\frac{2m-1}{3}\rfloor}-q^{\lfloor\frac{m-1}{3}\rfloor}-1$,
the sequence of $\delta_2$ is
\begin{equation*}
\begin{split}
\overline{\delta_2}=(\mathop{q-2}_{m-1},\underbrace{q-1, \ldots, q-1,}_{\frac{m}{3}-1} \mathop{q-2}_{\frac{2m}{3}-1},
\underbrace{q-1, \ldots, q-1,}_{ \frac{m}{3}-1} \mathop{q-2}_{\frac{m}{3}-1},
\underbrace{q-1, \ldots, q-1}_{ \frac{m}{3}-1})
\end{split}
\end{equation*}
if $m \equiv 0 \pmod{3}$ and
\begin{equation*}
\begin{split}
\overline{\delta_2}=(\mathop{q-2}\limits_{m-1},\underbrace{q-1, \ldots, q-1,}_{ \frac{m-4}{3}} \mathop{q-2}_{\frac{2m-2}{3}}, \underbrace{q-1, \ldots, q-1,}_{ \frac{m-1}{3}} \mathop{q-2}\limits_{\frac{m-4}{3}},
\underbrace{q-1, \ldots, q-1}_{\frac{m-4}{3}})
\end{split}
\end{equation*}
if $m\equiv 1 \pmod{3}$.

From the sequences of $\delta_3$, $\delta_2$ and Lemma \ref {lem:5}, it is easy to see that the sequence of the unique coset leader in the range $[\delta_3-1, \delta_2-1]$ is
\begin{equation*}
\begin{split}
(\mathop{q-2}_{m-1},\underbrace{q-1, \ldots, q-1,}_{\frac{m}{3}-2} \mathop{q-2}_{\frac{2m}{3}},
\underbrace{q-1, \ldots, q-1}_{ \frac{2m}{3}})
\end{split}
\end{equation*}
if $m \equiv 0 \pmod{3}$ and
\begin{equation*}
\begin{split}
(\mathop{q-2}\limits_{m-1},\underbrace{q-1, \ldots, q-1,}_{ \frac{m-7}{3}} \mathop{q-2}_{\frac{2m+1}{3}},
\underbrace{q-1, \ldots, q-1}_{ \frac{2m+1}{3}})
\end{split}
\end{equation*}
 if $m\equiv 1 \pmod{3}$. However, from the sequence we know that the unique coset leader is even. Hence,
there does not exist an odd coset leader in the range $[\delta_3-1, \delta_2-1]$. With an analysis similar as in the case 
$m\geq 6$, when $m=3$ and $m=4$, we also can obtain the desired results. This completes the proof.
\end{proof} 

With the help of Lemmas \ref{lem:2} and \ref{lem:21}, we can obtain the following theorem.

\begin{theorem}\label{thm:22}
 Let $n=\frac{q^m-1}2$, where $m\geq 2$ and $q$ is an odd prime power. Let $\delta_1$, $\delta_2$ and $\delta_3$ be given in Lemma \ref{lem:17}, Lemma \ref{lem:19} and Lemma \ref{lem:21}, respectively. Let $\delta$ be an integer. Then the negacyclic BCH code $\mathcal{C}_{(q, n,\delta,0)}$ has parameters
\begin{equation*}
\begin{cases}
	[\frac{q^m-1}{2},m,\frac{\delta_1+1}{2}] &{\rm if}~\frac{\delta_2+3}{2}\leq \delta \leq \frac{\delta_1+1}{2},\\
	[\frac{q^m-1}{2},m+\kappa,d\geq \frac{\delta_2+1}{2}] &{\rm if}~\frac{\delta_3+3}{2}\leq \delta \leq \frac{\delta_2+1}{2}; 
\end{cases}
\end{equation*}
furthermore, if $q^m\geq 25$, the negacyclic BCH code $\mathcal{C}_{(q, n,(\delta_3+1)/2,0)}$ has parameters $[\frac{q^m-1}{2},2m+\kappa,d\geq \frac{\delta_3+1}{2}] $, where $\kappa =\frac{m}3$ if $m\equiv 0 \pmod{3}$, and $\kappa =m$ if $m\not \equiv 0 \pmod{3}$.
\end{theorem}

\begin{proof}
	If $\frac{\delta_2+3}2\leq \delta \leq \frac{\delta_1+1}2$, then $\C_{(q, n, \delta, 0)}$ is the irreducible negacyclic code of length $n$ over $\gf(q)$ with check polynomial $\m_{\beta^{\delta_1}}(x)$. Then the trace expression of $\C_{(q, n, \delta, 0)}$ is given by
	   $$\C_{(q, n, \delta, 0)}=\{\bc(a)=(\tr_{q^m/q}(a \theta^{i}))_{i=0}^{n-1}:a\in \gf(q^m)\}, $$
	   where $\theta=\beta^{-\delta_1}=\beta^{q^{m-1}}$. Clearly, $\theta$ is a primitive $(q^m-1)$-th root of unity. For $a\in \gf(q^m)^*$,
	   \begin{align*}
	   	\wt(\bc(a))&=n-\frac{1}q\sum_{i=0}^{n-1}\sum_{x\in \gf(q)}\zeta_p^{\tr_{q/p}(x \tr_{q^m/q}(a\theta^i) )}\\
	   	&=n-\frac{1}{2q}\sum_{i=0}^{2n-1}\sum_{x\in \gf(q)}\zeta_p^{\tr_{q/p}(x \tr_{q^m/q}(a\theta^i) )}\\
	   	&=n-\frac{1}{2q}\sum_{x\in \gf(q)}\sum_{y\in \gf(q^m)^*} \zeta_p^{\tr_{q^m/p}(xay) }\\
	   	&=\frac{(q-1)n}{q}+\frac{q-1}{2q}-\frac{1}{2q}\sum_{x\in \gf(q)^*}\sum_{y\in \gf(q^m)} \zeta_p^{\tr_{q^m/p}(xay) }\\
	   	&=\frac{(q-1)q^{m-1}}2=\frac{\delta_1+1}2.
	   \end{align*}
	   where $p={\rm Char}(\gf(q))$ and $\zeta_p$ is a primitive $p$-th root of unity. Hence, $\C_{(q, n,\delta, 0)}$ is an $[n, m, \frac{\delta_1+1}2]$ one-weight code.
	   
	  If $\frac{\delta_3+3}2\leq \delta \leq \frac{\delta_2+1}2$, then $\C_{(q, n, \delta, 0)}$ is the negacyclic code of length $n$ over $\gf(q)$ with check polynomial $\m_{\beta^{\delta_1}}(x)\m_{\beta^{\delta_2}}(x)$. From Lemmas \ref{lem:2}, \ref{lem:17}, \ref{lem:19} and \ref{lem:20}, the desired conclusion then follows.
	  
	   If $\delta=\frac{\delta_3+1}2$, then $\C_{(q, n, \delta, 0)}$ is the cyclic code of length $n$ over $\gf(q)$ with check polynomial $\m_{\beta^{\delta_1}}(x)\m_{\beta^{\delta_2}}(x)\m_{\beta^{\delta_3}}(x)$. From Lemmas \ref{lem:2}, \ref{lem:17}, \ref{lem:19}, \ref{lem:20} and \ref{lem:21}, the desired conclusion then follows.   
\end{proof}

\begin{example}\label{example-02}
We have the following examples of the code of Theorem \ref{thm:22}.
\begin{itemize}
\item Let $q=3$, $m=5$ and $67\leq \delta \leq 81$, then the code $\mathcal{C}_{(3,121,\delta,0)}$ has parameters $[121,5,81]$.
\item Let $q=3$, $m=5$ and $ \delta = 63$, then the code $\mathcal{C}_{(3,121,\delta,0)}$ has parameters $[121,15, 63]$.
\item Let $q=3$, $m=4$ and $22\leq \delta \leq 27$, then the code $\mathcal{C}_{(3,40,\delta,0)}$ has parameters $[40,4,27]$.
\item Let $q=3$, $m=4$ and $14\leq \delta \leq 21$, then the code $\mathcal{C}_{(3,40,\delta,0)}$ has parameters $[40,8, 21]$.
\item Let $q=5$, $m=3$ and $48\leq \delta \leq 50$, then the code $\mathcal{C}_{(5,62,\delta,0)}$ has parameters $[62,3,50]$.
\end{itemize}
All the five codes are distance-optimal according to the tables of best codes known  in \cite{Grassl2006}.
\end{example}

\section{Negacyclic BCH codes with length $\frac{q^m+1}{2}$}\label{sec-MDSqodd}

Throughout this section, we always let $m\geq2$, $n=\frac{q^m+1}{2}$ and study the negacyclic BCH codes over $\gf(q)$ with length $n$. In this section, whenever we say ``$x$ is a coset leader", we mean that ``$x$ is a coset leader modulo $q^m+1$". 
That is to say, we omit the phrase ``modulo $q^m+1$".
 
 We first settle the parameters of this class of negacyclic BCH codes with some large dimensions. In fact, this work has been done when the designed distance of the code $\C_{(q,n,\delta,0)}$ is in the range $2\leq \delta \leq \frac{q^{\lfloor \frac{m-1}{2}\rfloor}+3}2$ in \cite{Pang18}. We here just give the parameters of $\C_{(q,n,\delta,0)}$ with the designed distance in a larger range. We start with the following result, which was given in \cite{Lid017,Liu17}.

\begin{lemma}\cite{Lid017,Liu17}
 Let $1\leq i\leq q^{m/2}$ if $m\geq 4$ is an even integer and $1\leq i< q^{(m+1)/2}-q+1$ if $m\geq 3$ is an odd integer. 
 Then $i$ is a coset leader if and only if $i \not\equiv 0 \pmod {q}$. Moreover, $|C_i^{(q,2n)}|=2m$.
\end{lemma}

With an analysis similar as in Theorem \ref{thm:15} and the help of Lemma \ref{lem:3}, we have the following theorem, 
which extends the work in \cite{Pang18}. 

\begin{theorem}\label{thm:25}
Let $2\leq \delta \leq \frac{q^{m/2}+3}{2}$ if $m\geq 4$ is an even integer and $2\leq \delta \leq \frac{q^{(m+1)/2}-q+2}{2}$ if $m\geq 3$ is an odd integer. Then the negacyclic BCH code $\C_{(q,n,\delta,0)}$ has parameters
$$\left[n,\, n-2m\left \lceil \frac{(2\delta-3)(q-1)}{2q} \right\rceil, \, d \right]$$
and generator polynomial
\begin{align*}
	g(x)=\prod_{\substack{0\leq i\leq \delta-2\\ i\not\equiv \frac{q-1}2 \pmod{q}}} \m_{\beta^{1+2i}}(x),
\end{align*}
where 
\begin{align*}
	d \geq \begin{cases}
	2\delta +1 &{\rm if~}\delta \equiv \frac{q+1}2 \pmod{q},  \\
	2\delta-1 &{\rm otherwise}.
\end{cases}
\end{align*}
\end{theorem}

\begin{example}\label{example-03}
We have the following examples of the code of Theorem \ref{thm:25}.
\begin{itemize}
\item Let $q=3$ and $m=5$, then the code $\mathcal{C}_{(3,122,3,0)}$ has parameters $[122,112, 5]$.
\item Let $q=5$ and $m=3$, then the code $\mathcal{C}_{(5,63,4,0)}$ has parameters $[63,51, 7]$.
\item Let $q=5$ and $m=3$, then the code $\mathcal{C}_{(5,63,5,0)}$ has parameters $[63,45,9]$.
\item Let $q=5$ and $m=2$, then the code $\mathcal{C}_{(5,13,4,0)}$ has parameters $[13,5,7]$.
\end{itemize}
All the four codes are distance-optimal according to the tables of best codes known  in \cite{Grassl2006} 
and their minimum distances achieve the lower bounds in Theorem \ref{thm:25}.  
\end{example}

In the rest of this section, we will study negacyclic BCH codes of length $n$ with some small dimensions. 
To this end, we need to find the first few largest odd coset leaders. From Lemma \ref{lem:7}, we can obtain 
the following results.

\begin{lemma}\label{lem:27}
Let $m\geq 2$ be an integer and $q$ be an odd prime power. Then the largest odd coset leader 
\begin{equation*}
\delta_1=\begin{cases}
n &{\rm if}~q^m\equiv 1 \pmod{4},\\
\frac{(q-1)n}{q+1}&{\rm if}~q^m\equiv 3 \pmod{4}.	
\end{cases}
\end{equation*}
Moreover, $|C_{\delta_1}^{(q,2n)}|=1$ if $q^m\equiv 1 \pmod{4}$ and $|C_{\delta_1}^{(q,2n)}|=2$ if $q^m\equiv 3 \pmod{4}$.
\end{lemma} 

\begin{proof}
From Lemma \ref{lem:7}, we know that the first largest coset leader is $n$. It is easily seen that $n$ is an odd integer if $q^m\equiv 1 \pmod{4}$ and $n$ is an even integer if $q^m\equiv 3 \pmod{4}$. Thus, $\delta_1=n$ if $q^m\equiv 1 \pmod{4}$.

It is obvious that $q^m\equiv 3 \bmod{4}$ if and only if $q\equiv 3 \pmod{4}$ and $m$ is an odd integer. Then
$$\frac{(q-1)n}{q+1}=\frac{q-1}{2}(q^{m-1}-q^{m-2}+\cdots-q+1),$$
which is an odd integer if $q^m\equiv 3 \pmod 4$. Hence, $\delta_1=\frac{(q-1)n}{q+1}$ from Lemma \ref{lem:7}. This completes the proof.
\end{proof}

In order to determine the values of $\delta_2$ and $\delta_3$, we recall a necessary and sufficient condition for $0\leq i\leq q^m$ being a coset leader, which was given in \cite{Zhu21}.

\begin{lemma}\cite[Proposition III.8]{Zhu21}\label{lem:28}
Let $q$ be an odd prime power. Let $1\leq j\leq m-1$, $l$ and $h$ be integers satisfying $1\leq l\leq \frac{q^j-1}{2}$ and 
$$-\frac{l(q^{m-j}-1)}{q^j+1}< h<\frac{l(q^{m-j}+1)}{q^j-1}.$$ Then $0\leq i \leq q^m$ is a coset leader if and only if $ i\leq n$ and $i \neq l q^{m-j}+h$.
\end{lemma}

\begin{lemma}\label{lem:29}
Let $1\leq i\leq n$ and $1\leq j\leq m-1$ be two integers and put 
$$a = iq^j \bmod {(q^m+1)}.$$
Then $i$ is not a coset leader if there exists a $j$ such that $1\leq a<i$ or $i+a>q^m+1$.
\end{lemma}

\begin{proof}
By definition, $i$ is not a coset leader if there exists a $j$ such that $1\leq a<i$. We now only show that $i$ is not a coset leader if there exists a $j$ such that $i+a>q^m+1$.

Let $iq^j=u(q^m+1)+a$ and let $l$ satisfy
\begin{equation}\label{eq:uab}
u+\frac{a-i}{q^m+1}<l<
u+\frac{a+i}{q^m+1}.
\end{equation}
Since $i+a>q^m+1$ and $i\leq n$, we have $a>i$. Then
$$u= \frac{iq^j-a}{q^m+1} \leq \left\lfloor \frac{i(q^j-1)}{q^m+1}\right\rfloor \leq \frac{q^j-1}{2}.$$
Hence, we have $1\leq l\leq \frac{q^j-1}{2}$ in (\ref{eq:uab}). From (\ref{eq:uab}) and $iq^j=u(q^m+1)+a$, we obtain that
\begin{equation*}
\frac{i(q^j-1)}{q^m+1}<l
<\frac{i(q^j+1)}{q^m+1}, 
\end{equation*}
which is the same as
\begin{equation*}
\frac{(q^m+1)l}{q^j+1}<i< \frac{(q^m+1)l}{q^j-1}.
\end{equation*}
Hence, $i$ can be written as $i=lq^{m-j}+h$, where
$$1\leq l\leq \frac{q^j-1}{2}\,\,
\text{and}\,\,-\frac{l(q^{m-j}-1)}{q^j+1}< h<\frac{l(q^{m-j}+1)}{q^j-1}.$$
From Lemma \ref{lem:28}, we deduce that $i$ is not a coset leader. This completes the proof.
\end{proof}

\begin{lemma}\label{lem:30}
Let $q^m\equiv 1 \pmod{4}$ and $\frac{q^m-1}2-q^{m-1}<i<n$ be odd. Assume that $i$ has the form $i=(\frac{q-3}2)q^{m-1}+i_{m-2}q^{m-2}+\cdots+i_0$, where $i_0,i_1,\ldots, i_{m-2}\in \{\frac{q-3}{2},\frac{q-1}{2},\frac{q+1}{2}\}$. Then $i$ is not a coset leader.
\end{lemma}

\begin{proof}
Since $q^m\equiv 1 \pmod{4}$, we deduce that $q\equiv 3 \pmod{4}$ and $m$ is an even integer, or $q\equiv 1 \pmod{4}$.
Let
$$e_0=\left| \left\{j: i_j=\frac{q+1}2,~{\rm where}~0\leq j\leq m-2 \right\}\right|$$
and
$$e_1=\left|\left\{j: i_j=\frac{q-3}2,~{\rm where}~ 0\leq j\leq m-2\right\} \right|.$$
If $q\equiv 1 \pmod{4}$, then $i_{m-2}q^{m-2}+\cdots+i_0$ is even since $i$ is odd. Hence, both $e_0$ and $e_1$ are odd, or both $e_0$ and $e_1$ are even since $i_0,i_1,\ldots, i_{m-2}\in \{\frac{q-3}{2},\frac{q-1}{2},\frac{q+1}{2}\}$. If $q\equiv 3 \pmod{4}$ and $m$ is even, then $i_{m-2}q^{m-2}+\cdots+i_0$ is odd since $i$ is odd. We can also deduce that both $e_0$ and $e_1$ are odd, or both $e_0$ and $e_1$ are even. Since the $q$-adic expansion of $\frac{q^m-1}{2}-q^{m-1}$ is $(\frac{q-3}2)q^{m-1}+(\frac{q-1}{2})(\sum_{j=0}^{m-2}q^j)$, the largest values of $0\leq j_1,j_2\leq m-2$ with $i_{j_1}=\frac{q+1}{2}$ and $i_{j_2}=\frac{q-3}{2}$ must satisfy $j_1>j_2$.

According to the analysis above, if $e_0=e_1$, then the sequence $(i_{m-2},i_{m-3},\ldots, i_1, i_0)$ must have one of the following forms:
\begin{itemize}
\item[(I)] $i_0=\frac{q-3}{2}$.
\item[(II)]  There exists $1\leq j_0 \leq m-3$ such that $i_{j_0-1}=\cdots=i_{0}=\frac{q-1}{2}$ and $i_{j_0}=\frac{q-3}{2}$.
\item[(III)] There exist $0\leq j_1<j_2\leq m-2$ such that $i_{j_1+1}=\cdots=i_{j_2-1}=\frac{q-1}{2}$ and $i_{j_1}=i_{j_2}=\frac{q-3}{2}$.
\item[(IV)] There exist $0\leq j_3<j_4\leq m-2$ such that $i_{j_3+1}=\cdots=i_{j_4-1}=\frac{q-1}{2}$ and $i_{j_3}=i_{j_4}=\frac{q+1}{2}$.
\end{itemize}
If $e_1>e_0$, then $e_1\geq e_0+2$ since $e_0$ and $e_1$ are both even or odd. It is clear that the sequence $(i_{m-2},i_{m-3},\ldots, i_1, i_0)$ must satisfy (III).
If $e_0>e_1$, then $e_0\geq e_1+2$. It is obvious that the sequence $(i_{m-2},i_{m-3},\ldots, i_1,i_0)$ must satisfy (IV).

To make use of Lemma \ref{lem:29}, let
\begin{equation}\label{LB}
a = iq^j \bmod{(q^m+1)}.
\end{equation}
In order to obtain the desired conclusion, it suffices to prove that there exists a $j$ such that $1\leq a<i$ or $i+a>q^m+1$. We now prove that this desired result holds in the following four cases.

 \noindent{\bf Case 1:}  $(i_{m-2}, i_{m-3}, \ldots, i_1,i_0)$ satisfies Condition (I). If $q=3$, i.e., $i_0=0$, then $q\mid i$. It follows that $i$ is not a coset leader. If $q>3$, let $j=m-1$ in (\ref{LB}), then
$$a=(\frac{q-3}{2})q^{m-1}-(\frac{q-3}{2})q^{m-2}-i_{m-2}q^{m-3}-\cdots-i_2q-i_1.$$
It is obvious that $1<a<i$.

\noindent{\bf Case 2:}  $(i_{m-2},i_{m-3},\ldots, i_1, i_0)$ satisfies Condition (II). In this case, let $j=m-j_0-1$ in (\ref{LB}),
then
\begin{align*}
a=&(\frac{q-3}{2})q^{m-1}+(\frac{q-1}{2})q^{m-2}+\cdots+(\frac{q-1}{2})q^{m-j_0-1}\\
&-(\frac{q-3}{2})q^{m-j_0-2}-i_{m-1}q^{m-j_0-3}-\cdots-i_{j_0+2}q-i_{j_0+1}.	
\end{align*}
It is clear that $1<a<\frac{q^m-1}{2}-q^{m-1}<i$.

\noindent{\bf Case 3:}  $(i_{m-2},i_{m-3},\ldots, i_1, i_0)$ satisfies Condition (III). Let $j=m-j_2-1$ in (\ref{LB}), then
\begin{equation*}
\begin{split}
a&=(\frac{q-3}{2})q^{m-1}+(\frac{q-1}{2})q^{m-2}+\cdots+(\frac{q-1}{2})q^{m-j_2+j_1}+\frac{q-3}{2}q^{m-j_2+j_1-1}\\
&+i_{j_1-1}q^{m-j_2+j_1-2}+\cdots+i_0q^{m-j_2-1}-(\frac{q-3}{2})q^{m-j_2-2}-i_{m-2}q^{m-j_2-3}-\cdots-i_{j_2+1}.
\end{split}
\end{equation*}
With an analysis similar as in Case 2, we have $1\leq i<a$.

\noindent{\bf Case 4:} $(i_{m-2}, i_{m-3},\ldots, i_1,i_0)$ satisfies Condition (IV). Let $j=m-j_4-1$ in (\ref{oddL-va}), we have
\begin{align*}
	a=&(\frac{q+1}2)q^{m-1}+(\frac{q-1}2)q^{m-2}+\cdots+(\frac{q-1}2)q^{m-j_4+j_3}+(\frac{q+1}2)q^{m-j_4+j_3-1}\\
	&+i_{j_3-1}q^{m-j_4+j_3-2}+\cdots+i_0q^{m-j_4-1}-(\frac{q-3}2)q^{m-j_4-2}-i_{m-2}q^{m-j_4-3}-\cdots-i_{j_4+1}.  
\end{align*}
Then
\begin{equation*}
\begin{split}
i+a&>\frac{q^m-1}{2}-q^{m-1}+a\\
&\geq(\frac{q-3}{2}+\frac{q+1}{2})q^{m-1}+(\frac{q-1}{2}+\frac{q-1}{2})q^{m-2}+\cdots+(\frac{q-1}{2}+\frac{q-1}{2})q^{m-j_4+j_3}\\
&+(\frac{q-1}{2}+\frac{q+1}{2})q^{m-j_4+j_3-1}+(\frac{q-1}{2})q^{m-j_4+j_3-2}+\cdots+ (\frac{q-1}{2})q^{m-j_4-1}\\
&+(\frac{q-1}{2}-\frac{q-3}{2})q^{m-j_4-2}+[\frac{q-1}{2}-(q-1)]q^{m-j_4-3}+\cdots+[\frac{q-1}{2}-(q-1)]\\
&>q^m+1.
\end{split}
\end{equation*}
Collecting the conclusions in all the cases above, we know that there always exists a $1\leq j\leq  m$ such that  $1\leq a<i$ or $i+a>q^m+1$. From Lemma \ref{lem:29}, the desired conclusion then follows.
\end{proof}

With the preparations above, we now determine the values of $\delta_2$ and $\delta_3$.

\begin{lemma}\label{lem:31}
Let $q$ be an odd prime power, then
\begin{equation*}
\delta_2=\begin{cases}
\frac{q^m-1}{2}-q^{m-1} &{\rm if}~q^m\equiv 1 \pmod{4},\\
\frac{(q-1)(q^m-2q^{m-2}-1)}{2(q+1)}&{\rm if}~q^m\equiv 3 \pmod{4}.
\end{cases}
\end{equation*}
Moreover, $|C_{\delta_2}^{(q,2n)}|=2m$.
\end{lemma}

\begin{proof} 
Note that $q^m\equiv 3 \pmod 4$ if and only if $q\equiv 3 \pmod 4$ and $m$ is odd. Thereby,
$$\frac{(q-1)(q^m-2q^{m-2}-1)}{2(q+1)}=\left(\frac{q-1}{2}\right)\left(q^{m-2}(q-1)-\frac{q^{m-2}+1}{q+1}\right)\equiv 1 \pmod 2.$$
It follows from Lemma \ref{lem:7} and Lemma \ref{lem:27} that $\delta_2=\frac{(q-1)(q^m-2q^{m-2}-1)}{2(q+1)}$ and $|C_{\delta_2}^{(q,2n)}|=2m$.

In the following, we only consider the case $q^m\equiv 1 \pmod 4$. We first prove that $\frac{q^m-1}{2}-q^{m-1}$ is a coset leader. From Lemma \ref{lem:28}, it suffices to show that there do not exist $1\leq j\leq m-1$, $1\leq l\leq \frac{q^j-1}{2}$ and $-\frac{l(q^{m-j}-1)}{q^j+1}< h<\frac{l(q^{m-j}+1)}{q^j-1}$ such that
\begin{equation}\label{eq:mqih}
\frac{q^m-1}{2}-q^{m-1}=lq^{m-j}+h.
\end{equation}
If (\ref{eq:mqih}) holds, then we have
\begin{equation*}
\frac{(q^m+1)l}{q^j+1}<\frac{q^m-1}{2}-q^{m-1}< \frac{(q^m+1)l}{q^j-1},
\end{equation*}
which is the same as
\begin{equation}\label{eq:fff-0T}
\frac{\left(\frac{q^m-1}{2}-q^{m-1}\right)(q^j-1)}{q^m+1}<l
<\frac{\left(\frac{q^m-1}{2}-q^{m-1}\right)(q^j+1)}{q^m+1}.
\end{equation}
For each $1\leq j\leq m-1$, suppose that $(\frac{q^m-1}2-q^{m-1})q^j=u_j(q^m+1)+r_j$, where $0\leq r_j\leq q^m$. It is clear that 
$$\left(\frac{q^m-1}2-q^{m-1}\right)q^j\equiv \frac{q^m+1}2-q^j+q^{j-1} \pmod{q^m+1}.$$
Hence, $r_j=\frac{q^m+1}2-q^j+q^{j-1}$. Furthermore, (\ref{eq:fff-0T}) becomes
$$u_j+\frac{q^{m-1}-q^j+q^{j-1}+1}{q^m+1}<l<u_j+\frac{q^{m}-q^{m-1}-q^j+q^{j-1}}{q^m+1}.$$
Hence, there does not exist an integer $l$  such that (\ref{eq:fff-0T}) holds. This means that  $\frac{q^m-1}{2}-q^{m-1}$ is a coset leader.

On one hand, $\frac{q^m-1}{2}-q^{m-1}<r_j<\frac{q^m+1}2$ for $1\leq j\leq m-1$, and
$$\left(\frac{q^m-1}2-q^{m-1}\right)q^{m}\equiv \frac{q^m+1}2+1+q^{m-1} \pmod{q^m+1}.$$
 It follows that $|C_{\frac{q^m-1}2-q^{m-1}}^{(q,2n)}|>m$. On the other hand, $|C_{\frac{q^m-1}2-q^{m-1}}^{(q,2n)}|$ divides $2m$. Therefore,  $$|C_{\frac{q^m-1}2-q^{m-1}}^{(q,2n)}|=2m.$$

We next prove that there does not exist an odd coset leader in the range $\left[\frac{q^m-1}{2}-q^{m-1}+2,\frac{q^m-3}{2}\right]$.
Let $1\leq \mu\leq \frac{q^{m-1}-1}{2}$, then all odd positive integers in the range $\left[\frac{q^m-1}{2}-q^{m-1}+2,\frac{q^m-3}{2}\right]$
can be expressed as $\frac{q^m-1}{2}-q^{m-1}+2\mu$.

From Lemma \ref{lem:28}, in order to show this result, it suffices to prove that there always exist $1\leq j\leq m-1$, $1\leq l\leq \frac{q^j-1}{2}$ and $-\frac{l(q^{m-j}-1)}{q^j+1}< h<\frac{l(q^{m-j}+1)}{q^j-1}$ such that
\begin{equation*}
\frac{q^m-1}{2}-q^{m-1}+2\mu=l q^{m-j}+h
\end{equation*}
for any $1\leq \mu\leq \frac{q^{m-1}-1}{2}$,
which is the same as
\begin{equation*}
\frac{(q^m+1)l}{q^j+1}<\frac{q^m-1}{2}-q^{m-1}+2\mu< \frac{(q^m+1)l}{q^j-1},
\end{equation*}
which is the same as
\begin{equation}\label{eq:qm-0T}
\frac{\left(\frac{q^m-1}{2}-q^{m-1}+2\mu\right)(q^j-1)}{q^m+1}<l<\frac{\left(\frac{q^m-1}{2}-q^{m-1}+2\mu\right)(q^j+1)}{q^m+1}.
\end{equation}
It is easy to check that $1\leq l\leq \frac{q^j-1}{2}$ if $l$ satisfies (\ref{eq:qm-0T}). Hence, in order to show that there does not exist an odd coset leader in the range  $\left[\frac{q^m-1}{2}-q^{m-1}+2,\frac{q^m-3}{2}\right]$, we only need to show that there exist $j$ and $l$ satisfying (\ref{eq:qm-0T}).

Let $0< L<q^m-1$ be a positive integer such that
\begin{equation}\label{oddL-va}
\left(\frac{q^m-1}{2}-q^{m-1}+2\mu \right)q^j=(q^m+1)f+L.
\end{equation}
 Then
(\ref{eq:qm-0T}) becomes
\begin{equation*}
f+\frac{L-(\frac{q^m-1}{2}-q^{m-1}+2\mu)}{q^m+1}<l<f+\frac{L+\frac{q^m-1}{2}-q^{m-1}+2\mu}{q^m+1}.
\end{equation*}
This means that there exist $j$ and $l$ satisfying (\ref{eq:qm-0T}) if and only if
$L< \frac{q^m-1}{2}-q^{m-1}+2\mu$ or $L+\frac{q^m-1}{2}-q^{m-1}+2\mu>q^m+1$.

Let the $q$-adic expansion of $\frac{q^m-1}{2}-q^{m-1}+2\mu$ be
$$b_{m-1}q^{m-1}+b_{m-2}q^{m-2}+\cdots+b_0,$$
where $0\leq b_0,b_1,\ldots, b_{m-1}\leq q-1$. From the range of the value of $\mu$, we get that $b_{m-1}=\frac{q-1}{2}$ or $b_{m-1}=\frac{q-3}{2}$. Let $b_{k_0}=\mbox{max}\{b_0,b_1,\ldots, b_{m-2}\}$ and $b_{k_1}=\mbox{min}\{b_0,b_1,\ldots, b_{m-2}\}$. We now prove that $\frac{q^m-1}{2}-q^{m-1}+2\mu$ is not a coset leader in the following four cases.

\noindent{\bf Case 1:} $b_{m-1}=\frac{q-1}{2}$. From $\frac{q^m-1}2-q^{m-1}+2\mu \leq \frac{q^m-3}2$, we have $b_{k_1}\leq \frac{q-3}2$. Then
\begin{equation*}
\begin{split}
&(b_{m-1}q^{m-1}+b_{m-2}q^{m-2}+\cdots+b_{k_1}q^{k_1}+\cdots+b_0)q^{m-k_1-1}\\
\equiv & b_{k_1}q^{m-1}+\cdots+b_0q^{m-k_1-1}-b_{m-1}q^{m-k_1-2}-\cdots-b_{k_1+1} \pmod{q^m+1}.
\end{split}
\end{equation*}
Note that
\begin{align*}
	&b_{k_1}q^{m-1}+\cdots+b_0q^{m-k_1-1}-b_{m-1}q^{m-k_1-2}-\cdots-b_{k_1+1}\\
	\geq & b_{k_1}q^{m-1}+\cdots+b_0q^{m-k_1-1}-(\frac{q-1}2)q^{m-k_1-2}-q^{m-k_1-2}+1\\
	\geq & (\frac{q-1}2)q^{m-k_1-2}+1,
\end{align*}
and 
\begin{align*}
	&b_{k_1}q^{m-1}+\cdots+b_0q^{m-k_1-1}-b_{m-1}q^{m-k_1-2}-\cdots-b_{k_1+1}\\
	\leq & (\frac{q-3}2)q^{m-1}+(q-1)(q^{m-2}+\cdots+q^{m-k_1-1})-(\frac{q-1}2)q^{m-k_1-2}\\
	< & (\frac{q-1}2)q^{m-1}<\frac{q^m-1}2-q^{m-1}+2\mu.
\end{align*}
Hence, $\frac{q^m-1}{2}-q^{m-1}+2\mu$ is not a coset leader.

\noindent{\bf Case 2:} $b_{m-1}=\frac{q-3}{2}$ and $b_{k_0}\geq\frac{q+3}{2}$. Let $j=m-k_0-1$ in (\ref{oddL-va}), then
$$L=b_{k_0}q^{m-1}+\cdots+b_0q^{m- k_0-1}-(\frac{q-3}{2})q^{m-k_0-2}-b_{m-2}q^{m-k_0-3}-\cdots-b_{k_0+2}q-b_{k_0+1}.$$
From the $q$-adic expansion of $\frac{q^m-1}{2}-q^{m-1}+2\mu$, it is easily seen that 
$$L+\frac{q^m-1}{2}-q^{m-1}+2\mu>q^m+1.$$ 
From Lemma \ref{lem:29}, we deduce that $\frac{q^m-1}{2}-q^{m-1}+2\mu$ is not a coset leader.

\noindent{\bf Case 3:} $b_{m-1}=\frac{q-3}{2}$ and $b_{k_1}< \frac{q-3}{2}$. Let $j=m-k_1-1$ in (\ref{oddL-va}), then
$$L=b_{k_1}q^{m-1}+\cdots+b_0q^{m-k_1-1}-(\frac{q-3}{2})q^{m-k_1-2}-b_{m-2}q^{m-k_1-3}-\cdots-b_{k_1+2}q-b_{k_1+1}.$$
It is obvious that $L<\frac{q^m-1}{2}-q^{m-1}$. From Lemma \ref{lem:29}, we deduce that $\frac{q^m-1}{2}-q^{m-1}+2\mu$ is not a coset leader.

\noindent{\bf Case 4:} $b_{m-1}=\frac{q-3}{2}$, $b_{k_0}<\frac{q+3}{2}$ and $b_{k_1}\geq \frac{q-3}{2}$. In this case,
$$b_0,b_1,\ldots, b_{m-2}\in \left\{\frac{q-3}{2},\frac{q-1}{2},\frac{q+1}{2}\right\}.$$
Then from Lemma \ref{lem:30}, we know that $\frac{q^m-1}{2}-q^{m-1}+2\mu$ is not a coset leader.

Summarizing the conclusions in the four cases above, we conclude that $\frac{q^m-1}{2}-q^{m-1}+2\mu$ is not a coset leader for any $1\leq \mu\leq \frac{q^{m-1}-1}{2}$. The desired conclusion then follows.
\end{proof}


\begin{lemma}\label{lem:32}
Let $q\geq 3$ be an odd prime power. Let $m$ be a positive integer such that $q^m\geq 25$. Then
\begin{equation*}
\delta_3=\begin{cases}
	\frac{q^m-1}{2}-q^{m-1}-q+1 &{\rm if}~q^m\equiv 1 \pmod 4,\\
	\frac{(q-1)(q^m-2q^{m-2}-1)}{2(q+1)}-(q-1)^2 &{\rm if}~q^m\equiv 3 \pmod{4}~{\rm and}~m\geq 5,\\
	\frac{(q-1)(q^3-2q-1)}{2(q+1)}-(q+1) &{\rm if}~q\equiv 3 \pmod{4}~{\rm and}~m=3.
\end{cases}
\end{equation*}
Moreover, $|C_{\delta_3}^{(q,2n)}|=2m$.
\end{lemma}

\begin{proof}
From \cite[Lemma 7]{Liuv2}, we know that the fourth largest coset leader is 
$$\frac{(q-1)(q^m-2q^{m-2}-1)}{2(q+1)}-(q-1)^2$$ 
if $m\geq 5$. With an analysis similar as in Lemma \ref{lem:31}, we deduce that $$\delta_3=\frac{(q-1)(q^m-2q^{m-2}-1)}{2(q+1)} -(q-1)^2$$ and $|C_{\delta_3}^{(q,2n)}|=2m$ if $q^m \equiv 3 \pmod 4$ and $m\geq5$. When $m=3$ and $q\equiv 3 \pmod 4$, the first three largest coset leaders are given in Lemma \ref{lem:7}. In addition, the authors of \cite{Liuv2} showed that the fourth, the fifth and the sixth largest coset leaders are $\frac{(q-1)(q^3-2q+1)}{2(q+1)}-1$, $\frac{(q-1)(q^3-2q+1)}{2(q+1)}-q$ and $\frac{(q-1)(q^3-2q+1)}{2(q+1)}-(q+1)$, respectively. Since the first, the fourth and the fifth coset leaders are even, we have $\delta_3=\frac{(q-1)(q^3-2q+1)}{2(q+1)}-(q+1)$ in this case. Moreover, the authors of \cite{Liuv2} proved that $|C_{\delta_3}^{(q,2n)}|=2m$ if $\delta$ takes on the value above.

For each $1\leq j\leq m-1$, suppose that $(\frac{q^m-1}2-q^{m-1}-q+1)q^j=u_j(q^m+1)+r_j$, where $0\leq r_j\leq q^m$. It is easy to verify that 
 \begin{align*}
r_j=\begin{cases}
	\frac{q^m-1}2+q^{j-1}-q^{j+1}+1 &{\rm if}~1\leq j\leq m-2,\\
	\frac{q^m-1}2+q^{m-2}+2 &{\rm if}~j=m-1.
\end{cases} 	
 \end{align*}
Therefore, $r_j-(\frac{q^m-1}{2}-q^{m-1}-q+1)>0$ and 
$$r_j+\left(\frac{q^m-1}{2}-q^{m-1}-q+1\right)<q^m+1$$
for any $1\leq j\leq m-1$. With an analysis similar as in Lemma \ref{lem:31}, we can deduce that $\frac{q^m-1}{2}-q^{m-1}-q+1$ is a coset leader, and $|C_{\frac{q^m-1}{2}-q^{m-1}-q+1}^{(q,2n)}|=2m$.

We next prove that there does not exist a coset leader in the range
\begin{equation}\label{eq:24}
\begin{split}
\left[\frac{q^m-1}{2}-q^{m-1}-q+3,\ \frac{q^m-1}{2}-q^{m-1}-2\right].
\end{split}
\end{equation}
Let $3\leq \mu\leq q-2$, then any positive integer in the range of (\ref{eq:24}) can be expressed as $\frac{q^m+1}{2}-q^{m-1}-\mu$. There are the following two cases.

\noindent{\bf Case 1}: $3\leq \mu \leq \frac{q-1}{2}$. 
It is easy to verify that
$$\left(\frac{q^m+1}2-q^{m-1}-\mu \right)q^{m-1}\equiv \frac{q^m+1}2+q^{m-2}-\mu q^{m-1} \pmod{q^m+1}$$ 
and 
$$ 0<\frac{q^m+1}2+q^{m-2}-\mu q^{m-1}< \frac{q^m+1}2-q^{m-1}-\mu.$$
Hence, $\frac{q^m+1}{2}-q^{m-1}-\mu$ is not a coset leader.

\noindent{\bf Case 2}: $\frac{q+1}2\leq \mu \leq q-2$. 
It is easy to verify that
$$\left(\frac{q^m+1}2-q^{m-1}-\mu \right)q^{2m-1}\equiv -\left[\frac{q^m+1}2+q^{m-2}-\mu q^{m-1}\right]  \pmod{q^m+1}$$ 
and 
$$ 0\leq -\left[\frac{q^m+1}2+q^{m-2}-\mu q^{m-1 }\right]< \frac{q^m+1}2-q^{m-1}-\mu.$$
Hence, $\frac{q^m+1}{2}-q^{m-1}-\mu$ is not a coset leader.

Summarising the conclusions in all the cases above, we conclude that $\frac{q^m+1}{2}-q^{m-1}-\mu$ is not a coset leader for any $3\leq \mu\leq q-2$. The desired conclusion then follows.
\end{proof}


From Lemmas \ref{lem:3} and \ref{lem:32}, we obtain the following theorem.

\begin{theorem}\label{thm:32}
 Let $n=\frac{q^m+1}2$, where $m\geq 2$ and $q$ is an odd prime power. Let $\delta_1$, $\delta_2$ and $\delta_3$ be given in Lemma \ref{lem:27}, Lemma \ref{lem:31} and Lemma \ref{lem:32}, respectively. Let $\delta$ be an integer. Then the negacyclic BCH code $\mathcal{C}_{(q,n,\delta,0)}$ has parameters $[n,2m(i-1)+\kappa,d\geq \delta_i]$ if $\frac{\delta_{i+1}+3}{2}\leq \delta \leq \frac{\delta_i+1}{2}~(i=1,2)$; and if $q^m\geq 25$, then the negacyclic BCH code $\mathcal{C}_{(q,n,(\delta_3+1)/2,0)}$ has parameters 
$[n,4m+\kappa,d\geq \delta_3]$, where 
 \begin{align*}
 	\kappa=\begin{cases}
	1 &{\rm if}~q^m\equiv 1 \pmod 4,\\
	2 &{\rm if}~q^m\equiv 3 \pmod 4.
\end{cases}
 \end{align*}
\end{theorem}

\begin{example}\label{example-04}
We have the following examples of the code of Theorem \ref{thm:32}.
\begin{itemize}
\item Let $q=5$, $m=2$ and $3\leq \delta \leq 4$, then the code $\mathcal{C}_{(3,121,\delta,0)}$ has parameters $[13,5,7]$.
\item Let $q=7$, $m=2$ and $7\leq \delta \leq 9$, then the code $\mathcal{C}_{(7,25,\delta,0)}$ has parameters $[25,5,17]$.
\end{itemize}
All the two codes are distance-optimal according to the tables of best codes known  in \cite{Grassl2006} 
and their minimum distances achieve the lower bound in Theorem \ref{thm:32}. 
\end{example}

\section{The negacyclic BCH codes with large dimensions}\label{sec-largedimes}

In this section, we study the parameters of negacyclic BCH codes of length $\frac{q^m+1}{2}$ and $\frac{q^m-1}{2}$  with one or two zeros, which have  large dimensions. We first give the parameters of negacyclic BCH codes with only one zero. 

\begin{theorem}\label{Theorem-05}
Let $n=\frac{q^m+1}{2}$, where $m\geq 2$ and $q$ is an odd prime power. Then the negacyclic BCH code $\C_{(q,n,2,0)}$ with generator polynomial $\m_{\beta}(x)$ has parameters $[n,n-2m,d]$, where
\begin{eqnarray*}
d=\begin{cases}
5 & {\rm if}~q=3, \\
3 &  {\rm if}~$m$ {\rm ~is ~odd~ and~} q>3,  \\
4 &  {\rm if} ~m {\rm ~is ~even ~and~}q>3.
\end{cases}
\end{eqnarray*}
\end{theorem}

\begin{proof}
The dimension of the code follows from the fact that $|C_1^{(q,2n)}|=2m$. We now prove the desired conclusions of the Hamming distance of the code. The proof will be carried out by distinguishing the following two cases.

\noindent{\bf Case 1:} $q=3$. Note that $-3,-1,1,3\in C_1^{(q,2n)}$. From Lemma \ref{lem:2}, we have $d\geq 5$. From the sphere-packing bound, we have $d\leq 6$. Assume that there exists a negacyclic BCH code over $\gf(3)$ with parameters $[n, n-2m,6]$. Applying Lemma \ref{bound2}, we have $q=3$, $n=\frac{3^m+1}{2}$, $t=n-5$, $r=2$, and
    $$3^{n-2m}\leq \frac{3^{n-1}}{1+2(n-1)^2},$$
    which is impossible if $m>1$. Hence, $d=5$.

\noindent{\bf Case 2:} $q>3$. From Lemma \ref{lem:2} and the sphere-packing bound, we deduce that
  \begin{equation}\label{eq:bound2}
   3\leq d \leq 4.
  \end{equation}

Let $y=\beta^{\frac{q^m+1}{q+1}}$, where $m$ is odd. It is easily seen that $y \in \gf(q^2)$ since $\beta$ is a primitive $(q^m+1)$-th root of unity in $\gf(q^{2m})$. Then there always exist $a_1,a_2,a_3\in \gf(q)^*$ such that
$$a_1+a_2y+a_3y^2=0.$$ It is easy to check that $y^2\neq y$ and $y^2\neq 1$, then there exists a codeword in $\C_{(q,n,2,0)}$ with Hamming weight $3$. Hence, $d=3$ if $m$ is odd.

If $m$ is even and $d=3$, then there exist $a_0,a_1\in \gf(q)^*$ such that
$$a_0+a_1\beta^{i_1}+\beta^{i_2}=0,$$
which is equivalent to
\begin{equation}\label{eq:25}
a_0+a_1\beta^{i_1}=-\beta^{i_2}.
\end{equation}
Raising both sides of (\ref{eq:25}) to the $(q^m+1)$-th power, we have
$$a_0^2+a_1^2+a_0a_1(\beta^{i_1}+\beta^{-i_1})=1,$$
which means that $\beta^{i_1}+\beta^{-i_1} \in \gf(q)$. It follows that $(\beta^{i_1}+\beta^{-i_1})^q=\beta^{i_1}+\beta^{-i_1}$, i.e., 
$$\beta^{-qi_1}\left(\beta^{(q-1)i_1}-1\right)\left(\beta^{(q+1)i_1}-1\right)=0.$$
It implies that $\beta^{(q-1)i_1}=1$ or $\beta^{(q+1)i_1}=1$. It is clear that
$$\gcd(q^m+1,q-1)=\gcd(q^m+1,q+1)=2.$$ 
Hence, we obtain $\beta^{2i_1}=1$, which is contradictory to the fact that $\beta$ is a primitive $(q^m+1)$-th root of unity in $\gf(q^{2m})$ and $1\leq i_1< n$. From (\ref{eq:bound2}), the desired conclusion then follows.
\end{proof}

\begin{remark}\label{Rem-optimalCodes}
In Theorem \ref{Theorem-05},	 we have the following two infinite families of optimal codes.
\begin{itemize}
\item When $q=3$ and $n=\frac{q^m+1}2$, the negacyclic BCH code $\C_{(q, n, 2, 0)}$ is distance-optimal and dimension-optimal with respect to the sphere-packing bound and has parameters $[n, n-2m, 5]$.  
\item When $q\geq 5$, $m$ is even and $n=\frac{q^m+1}2$, the negacyclic BCH code $\C_{(q, n, 2, 0)}$ is distance-optimal with respect to the sphere-packing bound and has parameters $[n, n-2m, 4]$.  	
\end{itemize}
\end{remark}

\begin{example}\label{example-05}
We have the following examples of the code of Theorem \ref{Theorem-05}.
\begin{itemize}
\item Let $q=9$ and $m=2$, then the code $\mathcal{C}_{(9,41,2,0)}$ has parameters $[41,37,4]$.
\item Let $q=3$ and $m=4$, then the code $\mathcal{C}_{(3,41,2,0)}$ has parameters $[41,33,5]$.
\item Let $q=7$ and $m=2$, then the code $\mathcal{C}_{(5,25,2,0)}$ has parameters $[25,21,4]$.
\end{itemize}
All the codes are distance-optimal according to the tables of best codes known in \cite{Grassl2006}. 
\end{example}

\begin{theorem}\label{Theorem-06}
Let $n=\frac{q^m-1}{2}$, where $m\geq 2$ and $q$ is an odd prime power. Then the negacyclic BCH code $\C_{(q, n, 2, 0)}$ with generator polynomial $\m_{\beta}(x)$ is distance-optimal with respect to the sphere-packing bound and has parameters $[n, n-m, d]$, where
\begin{eqnarray*}
d=\begin{cases}
3 &{\rm  if~}q=3,  \\
2 &{\rm if~}q\neq 3.  \\
\end{cases}
\end{eqnarray*}
\end{theorem}

\begin{proof}
The dimension of the code follows from the fact $|C_1^{(q,2n)}|=m$.
 We now settle the Hamming distance of the code. Recall that $\beta$ is a primitive $(q^m-1)$-th root of unity in $\gf(q^m)$.
It is clear that $\beta^{\frac{q^m-1}{q-1}}\in \gf(q)$, then there exist $a_1,a_2\in \gf(q)^*$ such that $a_1+a_2\beta^{\frac{q^m-1}{q-1}}=0$. This implies that $d=2$ if $q>3$.

If $q=3$, note that $1, 3\in C_{1}^{(q,2n)}$, it then follows from Lemma \ref{lem:2} that $d\geq 3$. When $m$ is odd, from Lemma \ref{lem:12} and \cite[Theorem 13]{Li2017}, we know that $d=3$. When $m$ is even, let $y=\beta^{\frac{3^m-1}{3^2-1}}\in \gf(3^2)$. Then there exist $a_1,a_2,a_3\in \gf(3)^*$ such that $a_1+a_2y+a_3y^2=0$. It is easy to check that $y^2\neq y$ and $y^2\neq 1$, then there exists a codeword in $\C_{(3,n,2,0)}$ with Hamming weight $3$. Hence, $d=3$. The desired conclusion then follows.
 \end{proof}

 \begin{example}\label{example-06}
We have the following examples of the code of Theorem \ref{Theorem-06}.
\begin{itemize}
\item Let $q=9$ and $m=2$, then the code $\mathcal{C}_{(9,41,2,0)}$ has parameters $[40,38,2]$.
\item Let $q=3$ and $m=3$, then the code $\mathcal{C}_{(3,13,2,0)}$ has parameters $[13,10,3]$.
\item Let $q=3$ and $m=4$, then the code $\mathcal{C}_{(3,40,2,0)}$ has parameters $[40,36,3]$.
\end{itemize}
All the codes are distance-optimal according to the tables of best codes known  in \cite{Grassl2006}. 
\end{example}

\begin{theorem}\label{Theorem-07}
Let $n=\frac{q^m-1}{2}$, where $m\geq 2$ and $q\geq 7$ is an odd prime power. Then the negacyclic BCH code $\C_{(q, n, 3, 0)}$ with generator polynomial $\m_{\beta}(x)\m_{\beta^3}(x)$ is almost distance-optimal with respect to the sphere-packing bound and has parameters $[n, n-2m, 3]$. 
\end{theorem}

\begin{proof}
It is easy to check that the generator polynomial of $\C_{(q, n, 3,0)}$ is $\m_{\beta}\m_{\beta^3}(x)$, and 
$$\deg(\m_{\beta}(x)\m_{\beta^3}(x))=2m.$$	
Hence, $\dim(\C_{(q, n, 3, 0)})=n-2m$. We next consider the Hamming distance of this code. From Lemma \ref{lem:2}, we have $d\geq 3$. Let $i_1=\frac{2n}{q-1}$ and $i_2=\frac{4n}{q-1}$. Consider the following system of equations:
\begin{eqnarray}\label{Eq:29}
\begin{cases}
1+a_1\beta^{i_1}+a_2\beta^{i_2}=0,  \\
1+a_1\beta^{3i_1}+a_2\beta^{3i_2}=0.
\end{cases}
\end{eqnarray}
Note that $\beta^{i_1}, \beta^{i_2}\in \gf(q)^*$. We have  
$$\left|\begin{array}{cc}
\beta^{i_1}& \beta^{i_2}\\
\beta^{3i_1}& \beta^{3i_2}	
\end{array}
  \right|=\beta^{\frac{10n}{q-1}}\left(\beta^{\frac{4n}{q-1}}-1 \right) \in \gf(q)^*.$$
Then the equations above have a unique solution $(a_1,a_2)\in (\gf(q)^*)^2$. Clearly, $1+a_1x^{i_1}+a_2x^{i_2}$ is a codeword of $\C_{(q, n, 3, 0)}$. Hence, $d\leq 3$. Consequently, $d=3$. According to the sphere-packing bound, the minimum distance of the $[n,n-2m]$ linear code over $\gf(q)$ is at most $4$. Therefore, $\C_{(q, n, 2, 0)}$ is almost distance-optimal. The desired conclusion then follows.
\end{proof}

\begin{remark}
	By definition, it is clear that $\C_{(3, n,2,0)}=\C_{(3,n,3,0)}$ and $\C_{(5, n,3,0)}=\C_{(5,n,4,0)}$. Hence, we only consider the case that $q\geq 7$ in Theorem \ref{Theorem-07}.
\end{remark}





\begin{theorem}\label{thm41}
 Let $q=5$, $m\geq 2$ and $n=\frac{q^m-1}2$. Then the negacyclic BCH code $\C_{(q, n, 4,0)}$ with generator polynomial $\m_1(x)\m_{3}(x)$ is distance-optimal with respect to the sphere-packing bound and has parameters $[n, n-2m, 4]$. 
 \end{theorem}

\begin{proof}
 The conclusions of generator polynomial and the dimension of $\C_{(q, n, 4,0)}$ are obvious. From Lemma \ref{lem:2}, we have $d\geq 4$. By the the sphere packing bound, the minimal distance $d\leq 4$. Hence, $d=4$ and $\C_{(q, n, 4, 0)}$ is distance-optimal. The desired conclusion then follows.
\end{proof}

\begin{example}\label{example-07}
 Let $q=5$ and $m=2$, then the code $\mathcal{C}_{(5,12,3,0)}$ has parameters $[12,8,4]$. It is an almost MDS code and
is distance-optimal according to the tables of best codes known  in \cite{Grassl2006}. 
\end{example}

\section{Summary and concluding remarks}\label{sec-finals}


The main contributions of this paper are the following:
\begin{itemize}
\item We determined the first, second and third largest odd coset leaders modulo $q^m-1$ and $q^m+1$. Based on these odd coset  leaders, we analysed the parameters of the negacyclic BCH code $\mathcal{C}_{(n,q,\delta,0)}$ for $\delta$ in some ranges, where $n=\frac{q^m-1}{2}$ or $n=\frac{q^m+1}{2}$ (see Theorem \ref{thm:15}, Theorem \ref{thm:22}, Theorem \ref{thm:25} and Theorem \ref{thm:32}).

\item We investigated the parameters of the negacyclic BCH code $\C_{(q,n,2,0)}$ of length $\frac{q^m+1}{2}$ and  the parameters of the negacyclic BCH code $\C_{(q,n,2,0)}$ and  $\C_{(q,n,3,0)}$ of length $\frac{q^m-1}{2}$ (see Theorem \ref{Theorem-05}, Theorem \ref{Theorem-06}, Theorem \ref{Theorem-07} and Theorem \ref{thm41}).

\item We presented three infinite families of optimal negacyclic codes (see Remark \ref{Rem-optimalCodes} and Theorem \ref{thm41}). 
\end{itemize} 

According to the tables of best codes known in~\cite{Grassl2006}, many of the negacyclic BCH codes presented in this paper are optimal
(see Examples \ref{example-01},  \ref{example-02},  \ref{example-03},  \ref{example-04},  \ref{example-05}, \ref{example-06} and \ref{example-07}). Similar to cyclic BCH codes, negacyclic BCH codes 
have good parameters in general. This justifies why negacyclic codes are interesting.  

Note that the class of negacyclic BCH codes of length $(q^m+1)/2$ treated in this paper are LCD codes. The dimensions of all the codes treated 
in this paper were settled and a lower bound on their minimum distances was derived. However, it is extremely difficult to determine the minimum distance of these codes in general. According to our Magma experiments, in most cases the lower
bound on their minimum distances is actually the minimum distance of the codes. It would be interesting to settle the minimum distance of these negacyclic codes.

\end{document}